\RequirePackage{etoolbox}
\csdef{input@path}{{style/}{graphics/}}
\documentclass[aap,MSNbibl,seceqn,dvips]{arximspdf}
\makeatletter
   \@ifpackageloaded{graphicx}{}{\usepackage{graphicx}}
\makeatother
\usepackage{mathrsfs}
\usepackage{accents}
\usepackage{graphicx}

\doi{10.1214/14-AAP1026} %kopijuoti is PTS
\volume{25}
\issue{3}
\pubyear{2015}
\firstpage{1383}
\lastpage{1419}

\makeatletter
\newcommand{\unde}{\underaccent{\bar}}
\newcommand{\AEE}{\mathrm{AE}}
\newtheorem{lemma}{Lemma}[section]
\newtheorem{theorem}{Theorem}[section]
\newproclaim{definition}{Definition}[section]
\newtheorem{proposition}{Proposition}[section]
\newproclaim{assumption}{Assumption}[section]
\newproclaim{remark}{Remark}
\makeatother

\begin{document}
\begin{frontmatter}

\title{Utility maximization with addictive consumption habit formation in incomplete semimartingale~markets\thanksref{T1}}
\runtitle{Habit formation}

\begin{aug}
\author[A]{\fnms{Xiang}~\snm{Yu}\corref{}\ead[label=e1]{xymath@umich.edu}}
\runauthor{X. Yu}
\affiliation{University of Michigan}
\address[A]{Department of Mathematics\\
University of Michigan\\
530 Church Street\\
Ann Arbor, Michigan 48109\\
USA\\
\printead{e1}} %adresu isvedimo komanda gale!
\end{aug}
\thankstext{T1}{This work is part of the author's Ph.D. dissertation at
the University of Texas at Austin and is
partially supported by NSF Grant under the
award number DMS-09-08441.}

% HISTORY:
\received{\smonth{7} \syear{2012}}
\revised{\smonth{11} \syear{2013}}

% ABSTRACT
%
\begin{abstract}
This paper studies the continuous time utility maximization problem on
consumption with addictive habit formation in incomplete semimartingale
markets. Introducing the set of auxiliary state processes and the
modified dual space, we embed our original problem into a
time-separable utility maximization problem with a shadow random
endowment on the product space $\mathbb{L}_{+}^{0}(\Omega\times
[0,T],\mathcal{O},\overline{\mathbb{P}})$. Existence and uniqueness of the
optimal solution are established using convex duality approach, where
the primal value function is defined on two variables, that is, the
initial wealth and the initial standard of living. We also provide
sufficient conditions on the stochastic discounting processes and on
the utility function for the well-posedness of the original
optimization problem. Under the same assumptions, classical proofs in
the approach of convex duality analysis can be modified when the
auxiliary dual process is not necessarily integrable.
\end{abstract}

% KEYWORDS
\end{frontmatter}

%s1 #&#
\section{Introduction}\label{sec1}

During the past decades, the time separable von Neu\-mann--Morgenstern
preferences on consumption have been observed to be inconsistent with
some empirical evidences. For instance, the well-known magnitude of the
equity premium ({Mehra and Prescott} \cite{Mehra}) cannot be
reconciled with the preference $\mathbb{E}[\int_{0}^{T}\mathit
{U}(t,c_t)\,dt]$ where the instantaneous utility function $\mathit{U}$ is
only \mbox{derived} from the consumption rate process. As an alternative
modeling tool, \textit{linear addictive} habit formation preference has
attracted a lot of attention and has been actively investigated in
recent years. This new preference is defined by $\mathbb{E}[\int_0^T\mathit{U}(t,c_t-Z(c)_t)\,dt]$, where $\mathit{U}\dvtx [0,T]\times
(0,\infty)\rightarrow\mathbb{R}$ and the additional accumulative
process $Z(c)_t$, called the \textit{habit formation} or \textit{the
standard of living} process, describes the consumption history impact.
To be more precise, $Z(c)_t$ is the solution of the following recursive
equation:
\begin{eqnarray*}
dZ(c)_t &=&\bigl(\delta_tc_t -
\alpha_tZ(c)_t\bigr)\,dt,
\\
Z(c)_0 &=&z,
\end{eqnarray*}
where the discounting factors $\alpha_t$ and $\delta_t$ are assumed to
be nonnegative optional processes, and the given real number $z\geq0$
is called the \textit{initial habit} or the \textit{initial standard of
living}. Moreover, the consumption habits are assumed to be \textit
{addictive} in the sense that $c_t\geq Z(c)_t$ for all time $t\in
[0,T]$. Compared to the time separable case, a small drop in
consumption may cause large fluctuation in consumption net of the
subsistence level due to the standard of living constraint. The habit
formation preference can possibly explain sizable excess returns on
risky assets in equilibrium models even for moderate values of the
degree of risk aversion. Based on this, a vast literature recommends
this time nonseparable preference as the new economic paradigm. We
refer readers to, for instance, {Constantinides} \cite
{constantinides1988habit}, {Samuelson} \cite{RePEccc} and
{Campbell and Cochrane} \cite{Campbell}.

The study of habit formation in modern economics dates back to
{Hicks} \cite{Hicks} in $1965$ and {Ryder and Heal} \cite{Heal}
in $1973$. More recently, there have been some important contributions
in complete It\^{o} processes markets; see {Detemple and
Zapatero} \cite{detemple91,detemple92}, {Schroder and
Skiadas} \cite{Schroder01072002}, {Munk} \cite{Munk},
{Detemple and Karatzas} \cite{Detemple2003265} and {Englezos and
Karatzas} \cite{Eng09}. Several pioneering work have derived the
explicit feedback form of the optimal policies under different
assumptions and market models. However, in the words by
{Englezos and Karatzas} \cite{Eng09}, ``The existence
of an optimal portfolio/consumption pair in an incomplete market is an
open question\,\ldots, and new methodologies are needed to handle the
problem.'' Therefore, in this paper, we are interested
in the general incomplete semimartingale framework and aim to prove the
existence and uniqueness of the optimal solution to this path dependent
optimization problem.

The convex duality approach plays an important role in the treatment of
utility maximization problems in incomplete markets. To list a very
small subset of the existing literature, we refer to Karatzas et al.
\cite{KKK}, Kramkov and Schachermayer \cite{kram99,kram03},
Cvitanic, Schachermayer and Wang~\cite{Cvitanic}, Karatzas and {\v
{Z}}itkovi{\'c} \cite{KarZit03}, Hugonnier and Kramkov \cite{kram04}
and {\v{Z}}itkovi{\'c} \cite{Zit02,Zit05}.

Typically, the critical step to build conjugate duality is to define
the dual space as a proper extension of space $\mathcal{M}$, which is
the set of equivalent local martingale measure density processes. The
first natural choice is the bipolar set of the space $\mathcal{M}$,
which is the smallest convex, closed and solid set containing the set
$\mathcal{M}$. Kramkov and Schachermayer \cite{kram99,kram03}
and {\v{Z}}itkovi{\'c} \cite{Zit02}, proved that this bipolar set can
be characterized as the solid hull of the set $\mathcal{Y}(y)$, which
is defined as the set of supermartingale deflators
\begin{eqnarray*}
\mathcal{Y}(y)&=& \bigl\{Y |  Y_{0}=y, Y_t>0, t\in[0,T]
\mbox{ and } XY=(X_{t}Y_{t})_{0\leq t\leq T}
\\
&&\hspace*{59pt} \mbox{is a supermartingale for each } X\in\mathcal{X}(x) \bigr\}.
\end{eqnarray*}
Here $\mathcal{X}(x)$ denotes the set of accumulated gains/losses
processes under some admissible portfolios with initial endowment less
than or equal to $x$. However, according to the definition of habit
formation process $Z(c)_t$, if we derive the naive dual problem using
the Legendre--Fenchel transform and the first order condition, we
arrive at
\begin{eqnarray*}
&& \mathop{{\inf}}_{y>0, Y\in\mathcal{Y}(y)}\mathbb{E} \biggl[  \int_{0}^{T}
\mathit{V} \biggl(Y_{t}+\delta_{t}\mathbb{E} \biggl[\int
_{t}^{T}e^{\int_{t}^{s}(\delta_{v}-\alpha_{v})\,dv}Y_{s}\,ds \Big|\mathcal
{F}_{t} \biggr] \biggr)\,dt \biggr]
\\
&&\qquad {} -z\mathbb{E} \biggl[\int_{0}^{T}e^{\int_{0}^{t}(\delta_{v}-\alpha
_{v})\,dv}Y_{t}\,dt
\biggr].
\end{eqnarray*}
The first mathematical difficulty is the extra integral $\mathbb
{E}
[\int_{0}^{T}e^{\int_{0}^{t}(\delta_{v}-\alpha_{v})\,dv}Y_{t}\,dt ]$,
from which we can see that $\mathcal{Y}(y)$ is not the appropriate
space to show the existence of the optimal dual solution. However, it
still reminds us to invoke the general treatment of random endowment
developed by Cvitanic, Schachermayer and Wang~\cite{Cvitanic}, Karatzas
and {\v{Z}}itkovi{\'c} \cite{KarZit03} and {\v{Z}}itkovi{\'c} \cite
{Zit05}. They proposed another extension of the set $\mathcal{M}$,
which is now considered as the set of equivalent local martingale
measures, to the set $\mathcal{D}$ of bounded finitely additive
measures. Nevertheless, their approach is inadequate to deal with the
first term of the dual problem, when the conditional integral part
$\mathbb{E} [\int_{t}^{T}e^{\int_{t}^{s}(\delta_{v}-\alpha
_{v})\,dv}Y_{s}\,ds |\mathcal{F}_{t} ]$ in the conjugate function
$\mathit{V}$ is taken into account.

In order to avoid the complexity of the path-dependence, we propose the
transform from the consumption rate process $c_{t}$ to the auxiliary
process $\tilde{c}_t=c_{t}-Z(c)_t$, so that the primal utility
maximization problem becomes time separable with respect to the process
$\tilde{c}_{t}$. This substitution idea from $c_{t}$ to $\tilde{c}_{t}$
appeared first in the market isomorphism result for complete markets by
Schroder and Skiadas \cite{Schroder01072002}. And meanwhile, for each
equivalent local martingale measure density process $Y\in\mathcal{M}$,
we define the auxiliary dual process $\Gamma_{t}$ exactly by
\[
\Gamma_{t}\triangleq Y_{t}+\delta_{t}\mathbb{E}
\biggl[\int_{t}^{T}e^{\int
_{t}^{s}(\delta_{v}-\alpha_{v})\,dv}Y_{s}\,ds
\Big|\mathcal{F}_{t} \biggr]\qquad \forall t\in[0,T].
\]
The dual problem can therefore be formulated in terms of auxiliary
process $\Gamma_{t}$ instead of $Y_{t}$ so that the path dependence of
$Y_{t}$ can also be hidden in the definition of process $\Gamma_{t}$.

By introducing the process $\tilde{w}_{t}=e^{\int_{0}^{t}(-\alpha
_{v})\,dv}$, one can shift the integral\break $\mathbb{E} [\int_{0}^{T}e^{\int_{0}^{t}(\delta_{v}-\alpha_{v})\,dv}Y_{t}\,dt ]$ to the
integral $\mathbb{E} [\int_{0}^{T}\tilde{w}_t\Gamma_{t}\,dt
]$ with
respect to its auxiliary process $\Gamma_{t}$. With the aid of this
equality, we can treat the extra exogenous random term $\tilde{w}_{t}$
as the shadow random endowment density process and define the dual
functional on the properly modified space of $\Gamma_{t}$ instead of
$Y_t$. By enlarging the effective domain of values for $x$ and $z$, the
original utility maximization problem with habit formation can be
embedded into the framework of Hugonnier and Kramkov \cite{kram04} as
an abstract time-separable utility maximization problem on the product
space.

On the other hand, we are facing some troubles in applying the
classical duality results since the auxiliary process $\Gamma_t$ is not
integrable. For instance, to show the existence of the dual optimizer,
the trick of applying the de la Vall\'{e}e--Poussin theorem in the
proof of Lemma~3.2 in Kramkov and Schachermayer \cite{kram99} does
not work. And the argument of contradiction in the proof of Lemma $1$
in Kramkov and Schachermayer \cite{kram03} using the subsequence
splitting lemma will also fail by observing that constants may not be
contained in the corresponding space. Therefore, we impose the
additional sufficient conditions on habit formation discounting factors
$\alpha_t$ and $\delta_t$; see Assumption \ref{asssss} to guarantee
the well-posedness of the primal optimization problem. We also ask for
reasonable asymptotic elasticity conditions on utility functions
$\mathit{U}$ both at $x\rightarrow0$ and $x\rightarrow\infty$ for the
validity of several key assertions of our main results to hold true.

The rest of this paper is organized as follows: Section \ref{section2} introduces the financial market and consumption\vspace*{2pt} habit
formation process. In Section \ref{section3}, we define the auxiliary
process\vspace*{2pt} space $\overline{\mathcal{A}}(x,z)$, the enlarged space
$\widetilde{\mathcal{A}}(x,z)$ and the auxiliary dual space $\widetilde{\mathcal
{M}}$. The original path-dependent utility maximization problem is
embedded into an abstract time separable optimization problem with the
shadow random endowment. Section \ref{section4} is devoted to the
formulation of the two-dimensional dual problem over the properly
enlarged dual space $\widetilde{\mathcal{Y}}(y,r)$ such that the shadow
random endowment part can be hidden, and our main results are stated at
the end. Section \ref{section5} contains detailed proofs of our main
theorems.

%s2 #&#
\section{Market model}\label{section2}

%s2.1 #&#
\subsection{The financial market model}
We consider a financial market with $d\in\mathbb{N}$ risky assets
modeled by a $d$-dimensional semimartingale
%
%e2.1 #&#
\begin{equation}
S=\bigl(S_{t}^{(1)},\ldots, S_{t}^{(d)}
\bigr)_{t\in[0,T]}
\end{equation}
on a given filtered probability space $(\Omega, \mathcal{F}, \mathbb
{F}=(\mathcal{F}_{t})_{0\leq t\leq T}, \mathbb{P})$, where the
filtration $\mathbb{F}$ satisfies the usual conditions, and the
maturity time is given by $T$. To simplify the notation, we take
$\mathcal{F}=\mathcal{F}_T$. We make the standard assumption that there
exists one riskless bond $S_{t}^{(0)}\equiv1$, $\forall t\in[0,T]$, which
amounts to considering $S_{t}^{(0)}$ as the num\'{e}raire asset.

The portfolio process $H=(H_{t}^{(1)},\ldots, H_{t}^{(d)} )_{t\in
[0,T]}$ is a predictable $S$-in\-tegrable process representing the number
of shares of each risky asset held by the investor at time $t\in
[0,T]$. The accumulated gains/losses process of the investor under his
trading strategy $H$ by time $t$ is given by
%
%e2.2 #&#
\begin{equation}
X_{t}^{H}=(H\cdot S)_{t}=\displaystyle\sum
\limits
_{k=1}^{d}\int_{0}^{t}H_{u}^{(k)}\,dS_{u}^{(k)},\qquad
t\in[0,T].
\end{equation}

%s2.2 #&#
\subsection{Admissible portfolios and consumption habit formation}
The portfolio process $(H_{t})_{t\in[0,T]}$ is called \textit
{admissible} if there exists a constant bound $a\in\mathbb{R}$ such
that $X_{t}^{H}\geq a$, a.s. for all $t\in[0,T]$.

Now, given the initial wealth $x>0$, the agent will also choose an
intermediate consumption plan during the whole investment horizon, and
we denote the consumption rate process by $c_{t}$. The resulting
self-financing \textit{wealth process} $(W_{t}^{x,H,c})_{t\in[0,T]}$ is
given by
%
%e2.3 #&#
\begin{equation}
W_{t}^{x,H,c}\triangleq x+(H\cdot S)_{t}-\int
_{0}^{t}c_{s}\,ds,\qquad t\in[0,T].
\end{equation}

Besides of the wealth process, as we defined in the \hyperref[sec1]{Introduction}, the
associated \textit{consumption habit formation} process $Z(c)_t$ is
given equivalently by the following exponentially weighted average of
agent's past consumption integral and the initial habit:
%
%e2.4 #&#
\begin{equation}
\label{Z} Z(c)_t=ze^{-\int_{0}^{t}\alpha_{v}\,dv}+\int_{0}^{t}
\delta _{s}e^{-\int
_{s}^{t}\alpha_{v}\,dv}c_{s}\,ds.
\end{equation}
Here discounting factors $\alpha_{t}$ and $\delta_{t}$ measure,
respectively, the persistence of the initial habits level and the
intensity of consumption history. In this paper, we shall be mostly
interested in the general case when discounting factors $\alpha_t$ and
$\delta_t$ are stochastic processes which are allowed to be unbounded.
However, for technical reasons, we will assume that $\int_0^{t}(\delta
_u-\alpha_u)\,du<\infty$ a.s. for each $t\in[0,T]$.

Throughout this paper, we make the assumption that the consumption
habit is addictive, that is, $c_{t}\geq Z(c)_{t}$, $\forall t\in
[0,T]$, which is to say, the investor's current consumption rate shall
never fall below \textit{the standard of living} process.

A consumption process $(c_{t})_{t\in[0,T]}$ is defined to be
$(x,z)$-\textit{financeable} if there exists an admissible portfolio
process $(H_{t})_{t\in[0,T]}$ such that $W_{t}^{x,H,c}\geq0$, a.s. for
$\forall t\in[0,T]$ and the addictive habit formation constraint
$c_{t}\geq Z(c)_{t}$, a.s. for $\forall t\in[0,T]$ holds. The class of
all $(x,z)$-financeable consumption rate processes will be denoted by
$\mathcal{A}(x,z)$, for $x>0$, $z\geq0$.

%s2.3 #&#
\subsection{Absence of arbitrage}
A probability measure $\mathbb{Q}$ is called the \textit{equivalent
local martingale measure} if it is equivalent to $\mathbb{P}$ and if
$X_{t}^{H}$ is a local martingale under $\mathbb{Q}$. We denote by
$\mathcal{M}$ the family of equivalent local martingale measures, and
in order to rule out the arbitrage opportunities in the market, we
assume that
%
%e2.5 #&#
\begin{equation}
\label{no-arbitrage} \mathcal{M}\neq\varnothing.
\end{equation}
See {Delbaen and Schachermayer} \cite{Schachermayer94} and
\cite
{schachermayer98} for comprehensive discussions on the topic of no
arbitrage.

Define the RCLL process $Y^{\mathbb{Q}}$ by
\[
Y_{t}^{\mathbb{Q}}=\mathbb{E} \biggl[\frac{d\mathbb{Q}}{d\mathbb
{P}} \bigg|
\mathcal{F}_{t} \biggr]
\]
for the $\mathbb{Q}\in\mathcal{M}$. Then $Y^{\mathbb{Q}}$ is called an
equivalent local martingale measure density and with a slight abuse of
notation, we denote $\mathcal{M}$ also as the set of all equivalent
local martingale density processes.

The celebrated optional decomposition theorem (see {Kramkov}
\cite{MR1402653}) enables us to characterize the $(x,z)$-financeable
consumption process by the following budget constraint condition:

%
%pr2.1 #&#
\begin{proposition}\label{thm1}
The process $(c_{t})_{t\in[0,T]}$ is $(x,z)$-financeable if and only
if $c_{t}\geq Z(c)_{t}$, $\forall t\in[0,T]$ and
%
%e2.6 #&#
\begin{equation}
\label{budgetconst} \mathbb{E} \biggl[\int_{0}^{T}c_{t}Y_{t}\,dt
\biggr] \leq x\qquad  \forall Y\in\mathcal{M}.
\end{equation}
\end{proposition}
%

%s2.4 #&#
\subsection{The utility function}
The individual investor's preference is represented by a utility
function $\mathit{U}\dvtx [0,T]\times(0, \infty)\rightarrow\mathbb{R}$,
such that, for every $x>0$, $\mathit{U}(\cdot,x)$ is continuous on
$[0,T]$, and for every $t\in[0,T]$, the function $\mathit{U}(t,\cdot)$
is strictly concave, strictly increasing, continuously differentiable
and satisfies the Inada conditions,
%
%e2.7 #&#
\begin{equation}
\label{Inada} \mathit{U}'(t,0)\triangleq\mathop{{
\lim}}_{x\rightarrow0}\mathit {U}'(t,x)=\infty, \qquad\mathit{U}'(t,
\infty)\triangleq\mathop{{\lim }}_{x\rightarrow\infty}\mathit{U}'(t,x)=0,
\end{equation}
where\vspace*{1pt} $\mathit{U}'(t,x)\triangleq\frac{\partial}{\partial x}\mathit
{U}(t,x)$. For each $t\in[0,T]$, we extend the definition of the
utility function by $\mathit{U}(t,x)=-\infty$ for all $x<0$, which is
equivalent to the addictive habit formation constraint $c_t\geq Z(c)_t$
when the utility function is defined on the difference between the
consumption rate process $c_t$ and the habit formation process
$Z(c)_t$.

According to these assumptions, the inverse $\mathit{I}(t,\cdot
)\dvtx \mathbb
{R}_{+}\rightarrow\mathbb{R}_{+}$ of the function $\mathit
{U}'(t,\cdot
)$ exists for every $t\in[0,T]$, and is continuous and strictly
decreasing. The convex conjugate of the utility function is defined by
\[
\mathit{V}(t,y)\triangleq\mathop{{\sup}}_{x>0}\bigl\{\mathit
{U}(t,x)-xy\bigr\},\qquad y>0.
\]

Following the asymptotic elasticity condition on utility functions
coined by Kramkov and Schachermayer \cite{kram99} (see also Karatzas
and {\v{Z}}itkovi\'{c} \cite{KarZit03}), we make assumptions on the
asymptotic behavior of $\mathit{U}$ at both $x=0$ and $x=\infty$ for
future purposes.

%
%as2.1 #&#
\begin{assumption}\label{ASSUV}
The utility function $\mathit{U}$ satisfies the reasonable asymptotic
elasticity condition both at $x=\infty$ and $x=0$, that is,
%
%e2.8 #&#
\begin{equation}
\label{assAEU} \AEE_{\infty}[\mathit{U}] =\mathop{{\lim\sup}}_{x\rightarrow\infty
}
\biggl(\sup_{t\in[0,T]}\frac{x\mathit{U}'(t,x)}{\mathit{U}(t,x)} \biggr)<1
\end{equation}
and
%
%e2.9 #&#
\begin{equation}
\label{assAEV} \AEE_{0}[\mathit{U}] =\mathop{{\lim\sup}}_{x\rightarrow0}
\biggl(\sup_{t\in
[0,T]}\frac{x\mathit{U}'(t,x)}{|\mathit{U}(t,x)|} \biggr)<\infty.
\end{equation}
Moreover, in order to get some inequalities uniformly in time $t$, we
shall assume
%
%e2.10 #&#
\begin{equation}
\label{up} \lim_{x\rightarrow\infty} \Bigl(\inf_{t\in[0,T]}
\mathit {U}(t,x) \Bigr)>0
\end{equation}
and
%
%e2.11 #&#
\begin{equation}
\label{down} \lim_{x\rightarrow0} \Bigl(\sup_{t\in[0,T]}
\mathit{U}(t,x) \Bigr)<0.
\end{equation}
\end{assumption}

%re1 #&#
\begin{remark}
Many well-known utility functions satisfy reasonable asymptotic
elasticity conditions (\ref{assAEU}) and (\ref{assAEV}), for
example, the discounted log utility function $\mathit
{U}(t,x)=e^{-\beta
t}\log(x)$ or discounted power utility function $\mathit
{U}(t,x)=e^{-\beta t}\frac{x^{p}}{p}$ ($p<1$ and $p\neq0$), for some
constant $\beta>0$. However, it is also easy to check that the utility
function $\mathit{U}(t,x)=-e^{{1}/{x}}$ does not satisfy condition
(\ref{assAEV}), and the utility function $\mathit{U}(t,x)=\frac
{x}{\log{x}}$ does not satisfy the condition (\ref{assAEU}).
\end{remark}

%
%re2 #&#
\begin{remark}
If the utility function satisfies the lower bound assumption $\mathop
{{\inf}}_{t\in[0,T]}\mathit{U}(t,0)>-\infty$, then condition (\ref{assAEV}) is automatically verified, and if the utility function
satisfies the upper bound assumption $\mathop{{\sup}}_{t\in
[0,T]}\mathit
{U}(t,\infty)<\infty$, condition (\ref{assAEU}) holds true.
\end{remark}

Next, some technical results give the equivalent characterizations of
reasonable asymptotic elasticity conditions (\ref{assAEU}) and
(\ref{assAEV}). The proof is based on the fact that $-V$ is a
concave function and on similar arguments in Lemma $6.3$ of Kramkov and
Schachermayer \cite{kram99}; see also Proposition $3.7$ of Karatzas and
{\v{Z}}itkovi\'{c} \cite{KarZit03}.

%
%le2.1 #&#
\begin{lemma}\label{Tech}
Let $\mathit{U}(t,x)$ be a utility function satisfying Assumption
\ref{ASSUV}. We have $\AEE_{0}[\mathit{U}]<\infty$ if and only if
$\AEE_{\infty
}[\mathit{V}]<1$, where we define
%
%e2.12 #&#
\begin{equation}
\label{assistTech} \AEE_{\infty}[\mathit{V}] =\mathop{{\lim\sup}}_{y\rightarrow\infty
}
\biggl(\sup_{t\in[0,T]}\frac{y\mathit{V}'(t,y)}{\mathit{V}(t,y)} \biggr)<1.
\end{equation}
In each of the subsequent assertions, the infimum of $\gamma_1>0$, for
which these assertions hold true, equals the reasonable asymptotic
elasticity $\AEE_{\infty}[\mathit{U}]$, and the infimum of $\gamma_2>0$
equals the reasonable asymptotic elasticity $\AEE_{\infty}[\mathit{V}]$.
\begin{longlist}[(iii)]
\item[(i)] There are $x_{0}>0$ and $y_0>0$ for all $t\in[0,T]$ s.t.
\begin{eqnarray*}
\mathit{U}(t,\lambda x)&<&\lambda^{\gamma_1}\mathit{U}(t,x)\qquad \mbox {for }
\lambda>1,x\geq x_{0};
\\
\mathit{V}(t,\lambda y)&>&\lambda ^{\gamma_2}
\mathit{V}(t,y)\qquad\mbox{for }\lambda>1, y\geq y_0.
\end{eqnarray*}
\item[(ii)] There are $x_{0}>0$ and $y_0>0$ for all $t\in[0,T]$ s.t.
\begin{eqnarray*}
\mathit{U}'(t,x) &<& \gamma_1\frac{\mathit{U}(t,x)}{x}\qquad\mbox
{for } x\geq x_{0};
\\
\mathit{V}'(t,y)&>& \gamma_2 \frac{\mathit{V}(t,y)}{y}\qquad \mbox{for } y\geq y_0.
\end{eqnarray*}
\item[(iii)] There are $x_0>0$ and $y_{0}>0$ for all $t\in[0,T]$ s.t.
\begin{eqnarray*}
\mathit{V}(t,\mu y)&<&\mu^{-{\gamma_1}/{(1-\gamma_1)}}\mathit {V}(t,y)\qquad\mbox{for } 0<\mu<1,
0<y\leq y_{0};
\\
\mathit{U}(t,\mu x)&>&\mu^{-{\gamma_2}/{(1-\gamma_2)}}\mathit {U}(t,x)\qquad\mbox{for } 0<\mu<1,
0<x\leq x_0.
\end{eqnarray*}
\item[(iv)] There are $x_0>0$ and $y_{0}>0$ for all $t\in[0,T]$ s.t.
\begin{eqnarray*}
-\mathit{V}'(t,y)&<& \biggl(\frac{\gamma_1}{1-\gamma_1} \biggr)
\frac{\mathit
{V}(t,y)}{y}\qquad\mbox{for } 0<y\leq y_{0};
\\
-\mathit{U}'(t,x)&>& \biggl(\frac{\gamma_2}{1-\gamma_2} \biggr)
\frac
{\mathit
{U}(t,x)}{x}\qquad\mbox{for } 0<x\leq x_0.
\end{eqnarray*}
\end{longlist}
\end{lemma}
%

%s3 #&#
\section{A new characterization of financeable consumption
processes}\label{section3}

%s3.1 #&#
\subsection{Functional set up}
In the spirit of {Bouchard and Pham} \cite{pham} who treated the
wealth dependent problem (see also {{\v{Z}}itkovi\'{c}} \cite
{Zit05} on consumption and endowment with stochastic clock), we denote
$\mathcal{O}$ as $\sigma$-algebra of optional sets relative to the
filtration $(\mathcal{F}_{t})_{t\in[0,T]}$, and let $d\overline{\mathbb
{P}}=dt\times d\mathbb{P}$ be the measure on the product space
$(\Omega
\times[0,T],\mathcal{O})$ defined as
%
%e3.1 #&#
\begin{equation}
\overline{\mathbb{P}}[A]=\mathbb{E}^{\mathbb{P}} \biggl[\int_{0}^{T}
\mathbf {1}_{A}(t,\omega)\,dt \biggr]\qquad\mbox{for } A\in\mathcal{O}.
\end{equation}
We denote by $\mathbb{L}^{0}(\Omega\times[0,T],\mathcal{O},\overline
{\mathbb
{P}})$ ($\mathbb{L}^{0}$ for short) the set of all random variables on
the product space $\Omega\times[0,T]$ with\vspace*{2pt} respect to the optional
$\sigma$-algebra $\mathcal{O}$ endowed with the topology of convergence
in measure $\overline{\mathbb{P}}$. And from now on, we shall identify the
optional stochastic process $(Y_t)_{t\in[0,T]}$ with the random
variable $Y\in\mathbb{L}^{0}(\Omega\times[0,T],\mathcal{O},\overline
{\mathbb
{P}})$. We also define the positive orthant $\mathbb{L}_{+}^{0}(\Omega
\times[0,T],\mathcal{O},\overline{\mathbb{P}})$ ($\mathbb{L}_{+}^{0}$ for
short) as the set of $Y=Y(t,\omega)\in\mathbb{L}^{0}$ such that
\[
Y\geq0,\qquad \overline{\mathbb{P}}\mbox{-a.s.}
\]

Endow $\mathbb{L}_{+}^{0}$ with the bilinear form valued in $[0,\infty
]$ as
\[
\langle X, Y\rangle=\mathbb{E} \biggl[\int_{0}^{T}X_{t}Y_{t}\,dt
\biggr]\qquad \mbox{for all } X,Y\in\mathbb{L}_{+}^{0}.
\]

%s3.2 #&#
\subsection{Path-dependence reduction by auxiliary processes}
At this point, we are able to define the set of all $(x,z)$-financeable
consumption rate processes as a set of random variables on the product
space $(\Omega\times[0,T],\mathcal{O},\overline{\mathbb{P}})$, and
Proposition \ref{thm1} states that
%
%e3.2 #&#
\begin{equation}
\mathcal{A}(x,z)= \bigl\{c\in\mathbb{L}_{+}^{0}\dvtx
c_t\geq Z(c)_t,\ \forall t\in[0,T]\mbox{ and }
\langle c,Y\rangle\leq x,\ \forall Y\in \mathcal{M} \bigr\}.\hspace*{-25pt}
\end{equation}

However, the family $\mathcal{A}(x,z)$ may be empty for some values
$x>0$, $z\geq0$. We shall restrict ourselves to the \textit{effective
domain} $\overline{\mathcal{H}}$ which is defined as the union of the
\textit
{interior} of set such that $\mathcal{A}(x,z)$ is not empty, and the
boundary $\{x>0, z=0\}$
%
%e3.3 #&#
\begin{equation}
\overline{\mathcal{H}}\triangleq\int \bigl\{(x,z)\in(0,\infty)\times [0,\infty )\dvtx
\mathcal{A}(x,z)\neq\varnothing \bigr\}\cup(0,\infty)\times\{0\}.
\end{equation}
From the definition, $\overline{\mathcal{H}}$ includes the special case of
zero initial habit, that is, $z=0$.

Before we state the next result, we shall first impose some additional
conditions on the discounting factors $\alpha_{t}$ and $\delta_{t}$,
which are essential for the well-posedness of the primal utility
optimization problem:

%as3.1 #&#
\begin{assumption}\label{asssss}
We assume that nonnegative optional processes $\alpha_{t}$ and $\delta
_{t}$ satisfy
%
%e3.4 #&#
\begin{equation}
\label{ass31} \mathop{{\sup}}_{Y\in\mathcal{M}}\mathbb{E} \biggl[\int
_{0}^{T}e^{\int
_{0}^{t}(\delta_{v}-\alpha_{v})\,dv}Y_{t}\,dt \biggr]<
\infty,
\end{equation}
and there exists a constant $\bar{x}>0$ such that
%
%e3.5 #&#
\begin{equation}
\label{ass41} \mathbb{E} \biggl[\int_{0}^{T}
\mathit{U}^{-}\bigl(t,\bar{x}e^{-\int
_{0}^{t}\alpha_{v}\,dv}\bigr)\,dt \biggr] <\infty.
\end{equation}
\end{assumption}

%
%re3 #&#
\begin{remark}
If stochastic discounting processes $\alpha_{t}$ and $\delta_{t}$ are
assumed to be bounded, conditions (\ref{ass31}) and (\ref{ass41})
will be satisfied. Condition (\ref{ass31}) is the well-known
super-hedging property of the random variable $\int_0^Te^{\int
_0^{t}(\delta_v-\alpha_v)\,dv}\,dt$ in the original market.
\end{remark}

%
%le3.1 #&#
\begin{lemma}
Under condition (\ref{ass31}), the effective domain $\overline{\mathcal
{H}}$ can be rewritten explicitly as
%
%e3.6 #&#
\begin{equation}
\qquad \overline{\mathcal{H}}= \biggl\{(x,z)\in(0,\infty)\times[0,\infty)\dvtx  x>z\sup
_{Y\in\mathcal{M}}\mathbb{E} \biggl[\int_0^Te^{\int_0^t(\delta
_v-\alpha
_v)\,dv}Y_t\,dt
\biggr] \biggr\}.
\end{equation}
\end{lemma}

The proof of the lemma is straightforward, and we refer to Yu \cite{Yu}
for details.

By choosing $(x,z)\in\overline{\mathcal{H}}$, we can now define the
preliminary version of our \textit{primal utility maximization
problem} as
%
%e3.7 #&#
\begin{equation}
\label{primalp} u(x,z)\triangleq\mathop{{\sup}}_{c\in\mathcal{A}(x,z)}\mathbb {E} \biggl[
\int_{0}^{T}\mathit{U}\bigl(t,c_{t}-Z(c)_t
\bigr)\,dt \biggr],\qquad (x,z)\in\overline{\mathcal{H}}.
\end{equation}

Now, for fixed $(x,z)\in\overline{\mathcal{H}}$, and each
$(x,z)$-financeable consumption rate process, we want to generalize the
Market Isomorphism idea by Schroder and Skiadas \cite{Schroder01072002}
in order to reduce the path-dependence structure. By introducing the
auxiliary process $\tilde{c}_{t}=c_{t}-Z(c)_t$, the auxiliary set of
$\mathcal{A}(x,z)$ is given as
%
%e3.8 #&#
\begin{equation}
\label{setbarA} \overline{\mathcal{A}}(x,z)\triangleq \bigl\{\tilde{c}\in\mathbb
{L}_{+}^{0}\dvtx  \tilde{c}_{t}=c_{t}-Z(c)_t,\
\forall t\in[0,T], c\in \mathcal {A}(x,z) \bigr\}.
\end{equation}

It is straightforward to verify the following lemma.

%le3.2 #&#
\begin{lemma}
For\vspace*{1.5pt} each fixed $(x,z)\in\overline{\mathcal{H}}$, there is one-to-one
correspondence between sets $\mathcal{A}(x,z)$ and $\overline{\mathcal
{A}}(x,z)$, and hence we have $\overline{\mathcal{A}}(x,z)\neq\varnothing$ for
$(x,z)\in\overline{\mathcal{H}}$.
\end{lemma}

Let us turn our attention to the set $\mathcal{M}$ of equivalent local
martingale measure densities, and for each $Y\in\mathcal{M}$, the
auxiliary optional process with respect to $Y_{t}$ is defined as
%
%e3.9 #&#
\begin{equation}
\label{Gamma} \Gamma_{t}\triangleq Y_{t}+
\delta_{t}\mathbb{E} \biggl[\int_{t}^{T}e^{\int
_{t}^{s}(\delta_{v}-\alpha_{v})\,dv}Y_{s}\,ds
\Big|\mathcal{F}_{t} \biggr]\qquad \forall t\in[0,T].
\end{equation}

Let us denote the set of all these auxiliary optional processes by
%
%e3.10 #&#
\begin{eqnarray}\label{settildeM}
\widetilde{\mathcal{M}} &=& \biggl\{\Gamma\in\mathbb{L}_{+}^{0}\dvtx
\Gamma _{t}=Y_{t}+\delta_{t}\mathbb{E} \biggl[
\int_{t}^{T}e^{\int
_{t}^{s}(\delta
_{v}-\alpha_{v})\,dv}Y_{s}\,ds \Big|
\mathcal{F}_{t} \biggr],
\nonumber\\[-8pt]\\[-8pt]
&&\hspace*{146pt} \forall t\in [0,T], Y\in\mathcal{M} \biggr\}.\nonumber
\end{eqnarray}
We remark here that although stochastic discounting processes $\delta
_{t}$ and $\alpha_{t}$ are unbounded, under condition (\ref{ass31}),
the auxiliary dual process $\Gamma$ is well defined in $\mathbb
{L}_{+}^{0}$, but it is not necessarily in $\mathbb{L}^{1}$.

A direct application of the Fubini--Tonelli theorem induces the key
equalities below; for the detailed proof, we refer to Proposition
$2.3.3$ of Yu \cite{Yu}.

%
%pr3.1 #&#
\begin{proposition}\label{prop31}
Under condition (\ref{ass31}), for each nonnegative optional process
$c_{t}$ such that $c_{t}\geq Z(c)_t$ with $Z(c)_t$ defined by (\ref{Z}) for fixed initial standard of living $z\geq0$ and the
nonnegative optional process $Y_{t}$, we have the following equalities
with respect to their corresponding auxiliary processes $\tilde
{c}_{t}=c_{t}-Z(c)_t$ and $\Gamma_{t}$ which is defined by (\ref{Gamma}), that
%
%e3.11 #&#
\begin{eqnarray}
\label{eqnequiv} {\langle} c,Y {\rangle} & =& {\langle}\tilde{c},\Gamma {\rangle}+z
{\langle} w,Y {\rangle}
\nonumber
\\[-8pt]
\\[-8pt]
& =& {\langle}\tilde{c},\Gamma {\rangle}+z {\langle} \tilde {w},\Gamma {\rangle}.
\nonumber
\end{eqnarray}
Here we define random processes $w, \tilde{w}\in\mathbb{L}_{+}^{0}$ by
%
%e3.12 #&#
\begin{equation}
\label{shadow} w_{t}\triangleq e^{\int_{0}^{t}(\delta_{v}-\alpha_{v})\,dv}\quad\mbox {and}\quad
\tilde{w}_{t}\triangleq e^{\int_{0}^{t}(-\alpha
_{v})\,dv}\qquad \mbox{for all } t
\in[0,T].
\end{equation}
\end{proposition}

Based on Propositions \ref{thm1} and \ref{prop31}, under
conditions (\ref{ass31}) and (\ref{ass41}), clearly we will have
the alternative budget constraint characterization of the consumption
rate process $c_{t}$ as:

%
%pr3.2 #&#
\begin{proposition}\label{char}
For any given pair $(x,z)\in\overline{\mathcal{H}}$, the consumption rate
process $c$ is $(x,z)$-financeable if and only if $c_{t}\geq Z(c)_t$,
$\forall t\in[0,T]$ and
\[
{\bigl\langle} c-Z(c),\Gamma {\bigr\rangle}\leq x-z {\langle} \tilde {w},\Gamma
{\rangle}\qquad\mbox{for all } \Gamma\in\widetilde {\mathcal{M}}.
\]
\end{proposition}

Proposition \ref{char} provides us the alternative definition of set
$\overline{\mathcal{A}}(x,z)$ for $(x,z)\in\overline{\mathcal{H}}$ by
%
%e3.13 #&#
\begin{equation}
\label{setbarAA} \overline{\mathcal{A}}(x,z)= \bigl\{\tilde{c}\in\mathbb{L}_{+}^{0}\dvtx
{\langle}\tilde{c},\Gamma {\rangle}\leq x-z {\langle} \tilde {w},\Gamma {
\rangle},\ \forall\Gamma\in\widetilde{\mathcal {M}} \bigr\}.
\end{equation}
%

%s3.3 #&#
\subsection{Embedding into an abstract utility maximization problem with the shadow random endowment}
In order\vspace*{2pt} to apply the convex duality approach for the random endowment,
we need to enlarge the\vspace*{1pt} domain of the set $\overline{\mathcal{H}}$ to
$\mathcal{H}$ and also enlarge the corresponding auxiliary set $\overline
{\mathcal{A}}(x,z)$ to $\widetilde{\mathcal{A}}(x,z)$,
%
%e3.14 #&#
\begin{equation}
\label{settildeA} \widetilde{\mathcal{A}}(x,z)\triangleq \bigl\{\tilde{c}\in\mathbb
{L}_{+}^{0}\dvtx  {\langle}\tilde{c},\Gamma {\rangle}\leq x-z {
\langle} \tilde{w},\Gamma {\rangle},\ \forall\Gamma\in \widetilde {\mathcal{M}}
\bigr\},
\end{equation}
where now $(x,z)\in\mathbb{R}^{2}$ and is restricted in the enlarged
domain $\mathcal{H}$,
\[
\mathcal{H}\triangleq \operatorname{int} \bigl\{(x,z)\in\mathbb{R}^{2}\dvtx  \widetilde {
\mathcal{A}}(x,z)\neq\varnothing \bigr\}.
\]

Under condition (\ref{ass31}) and Proposition \ref{prop31}, it is
easy to verify the equivalent characterization of $\mathcal{H}$ by the
following:

%
%le3.3 #&#
\begin{lemma}\label{lemKKK}
%
%e3.15 #&#
\begin{eqnarray}
\label{setK} \mathcal{H}&= & \bigl\{(x,z)\in\mathbb{R}^{2}\dvtx  x> z {
\langle} \tilde {w}, \Gamma {\rangle}, \mbox{ for all } \Gamma\in \widetilde {
\mathcal{M}} \bigr\}
\nonumber
\\[-8pt]
\\[-8pt]
&= & \bigl\{(x,z)\in\mathbb{R}^{2}\dvtx  x>\bar{p}z, z\geq0 \bigr\}\cup
\bigl\{ (x,z)\in\mathbb{R}^{2}\dvtx x>\unde{p}z, z<0 \bigr\}.
\nonumber
\end{eqnarray}
Here
%
%e3.16 #&#
\begin{equation}
\label{pbar} \bar{p}\triangleq\mathop{{\sup}}_{Y\in\mathcal{M}} {\langle }w,Y {
\rangle}=\mathop{{\sup}}_{\Gamma\in\widetilde{\mathcal{M}}} {\langle }\tilde{w},\Gamma {\rangle},
\end{equation}
and
%
%e3.17 #&#
\begin{equation}
\label{pdown} \unde{p}\triangleq\mathop{{\inf}}_{Y\in\mathcal{M}} {\langle
}w,Y {\rangle}=\mathop{{\inf}}_{\Gamma\in\widetilde{\mathcal
{M}}} {\langle}\tilde{w},\Gamma {
\rangle},
\end{equation}
where $\bar{p},\unde{p}<\infty$, and $\mathcal{H}$ is a
well-defined convex cone in $\mathbb{R}^{2}$. Moreover,
%
%e3.18 #&#
\begin{eqnarray}
\label{closeK} \operatorname{cl}\mathcal{H}&= & \bigl\{(x,z)\in\mathbb{R}^{2}\dvtx
\widetilde{\mathcal {A}}(x,z)\neq\varnothing \bigr\}
\nonumber
\\[-8pt]
\\[-8pt]
&= & \bigl\{(x,z)\in\mathbb{R}^{2}\dvtx  x\geq z {\langle} \tilde{w},
\Gamma {\rangle},\mbox{ for all } \Gamma\in\widetilde{\mathcal {M}} \bigr\},
\nonumber
\end{eqnarray}
where $\operatorname{cl}\mathcal{H}$ denotes the closure of the set
$\mathcal{H}$ in
$\mathbb{R}^{2}$.
\end{lemma}

We will now define the \textit{auxiliary primal utility maximization
problem} based on the auxiliary domain $\widetilde{\mathcal{A}}(x,z)$ as
%
%e3.19 #&#
\begin{equation}
\label{tilde-u} \tilde{u}(x,z)\triangleq\mathop{{\sup}}_{\tilde{c}\in\widetilde
{\mathcal {A}}(x,z)}\mathbb{E}
\biggl[\int_{0}^{T}\mathit{U}(t,\tilde
{c}_{t})\,dt \biggr],\qquad (x,z)\in\mathcal{H}.
\end{equation}

By definitions of $\overline{\mathcal{A}}(x,z)$ for $(x,z)\in\overline
{\mathcal
{H}}$ and $\widetilde{\mathcal{A}}(x,z)$ for $(x,z)\in\mathcal{H}$, we
successfully embedded our original utility maximization problem (\ref
{primalp}) with consumption habit formation into the auxiliary
abstract utility maximization problem (\ref{tilde-u}) without habit
formation, but with the shadow random endowment. More precisely, the
following equivalence can be guaranteed that for any $(x,z)\in\overline{\mathcal{H}}\subset\mathcal{H}$,
%
%e3.20 #&#
\begin{equation}
\label{barA=tildeA} \overline{\mathcal{A}}(x,z)=\widetilde{\mathcal{A}}(x,z),
\end{equation}
and the two value functions coincide,
%
%e3.21 #&#
\begin{equation}
\label{u=tilde-u} u(x,z)=\tilde{u}(x,z).
\end{equation}
In addition, we have that $c_{t}^{\ast}$ is the optimal solution for
$u(x,z)$ if and only if $\tilde{c}_{t}^{\ast}=c_{t}^{\ast}-Z(c^{\ast
})_{t}\geq0$, and for all $t\in[0,T]$ is the optimal\vspace*{1.5pt} solution for
$\tilde{u}(x,z)$, when $(x,z)\in\overline{\mathcal{H}}$.

%s4 #&#
\section{The dual optimization problem and main results}\label{section4}
Inspired by the idea in Hugonnier and Kramkov \cite{kram04} for
optimal investment with random endowment, we concentrate now on the
construction of the dual problem by first introducing the set $\mathcal{R}$,
%
%e4.1 #&#
\begin{equation}
\label{setL} \mathcal{R}\triangleq ri \bigl\{(y,r)\in\mathbb{R}^{2}\dvtx
xy-zr\geq0,\mbox{ for all } (x,z)\in\mathcal{H} \bigr\}.
\end{equation}

Let us make the following assumption on stochastic discounting processes
$\alpha_{t}$ and~$\delta_{t}$.

%as4.1 #&#
\begin{assumption}\label{ass5}
The random variable defined by
%
%e4.2 #&#
\begin{equation}
\mathcal{E}\triangleq\int_{0}^{T}w_{t}\,dt=
\int_{0}^{T}e^{\int
_{0}^{t}(\delta_{v}-\alpha_{v})\,dv}\,dt
\end{equation}
is not replicable under our original financial market; that is, there
is no constant $K$ such that
\[
\mathbb{E}^{\mathbb{Q}}[\mathcal{E}]=K\qquad\mbox{for any } \mathbb {Q}\in
\mathcal{M}.
\]
\end{assumption}

%
%re4 #&#
\begin{remark}
Under Assumption \ref{ass5}, the existence of market isomorphism by
Schroder and Skiadas \cite{Schroder01072002} may no longer hold. Our
work can generally extend their conclusions and provide the existence
and uniqueness of the optimal solution in incomplete markets using
convex analysis.
\end{remark}

%
%le4.1 #&#
\begin{lemma}\label{lemL}
Under Assumption \ref{ass5}, we know that $\mathcal{R}$ is an open
convex cone in $\mathbb{R}^{2}$ and can be rewritten as
%
%e4.3 #&#
\begin{equation}
\label{LLL} \mathcal{R}= \bigl\{(y,r)\in\mathbb{R}^{2}\dvtx  y>0\mbox{ and } \unde{p}y<r<\bar{p}y \bigr\},
\end{equation}
where $\bar{p}$ and $\unde{p}$ are defined by (\ref{pbar}) and
(\ref{pdown}), and $\bar{p}<\unde{p}$.
\end{lemma}

For\vspace*{2pt} an arbitrary pair $(y,r)\in\mathcal{R}$, we denote by $\widetilde
{\mathcal{Y}}(y,r)$ the set of nonnegative processes as a proper
extension of the auxiliary set $\widetilde{\mathcal{M}}$ in the way that
%
%e4.4 #&#
\begin{equation}
\label{setDual} \widetilde{\mathcal{Y}}(y,r)\triangleq \bigl\{\Gamma\in\mathbb
{L}_{+}^{0}\dvtx  {\langle}\tilde{c},\Gamma {\rangle}\leq xy-zr,
\mbox{ for all } \tilde{c}\in\widetilde{\mathcal{A}}(x,z)\mbox{ and } (x,z)\in
\mathcal{H} \bigr\}.\hspace*{-25pt}
\end{equation}

The \textit{auxiliary dual utility maximization problem} to (\ref{tilde-u}) can be now defined by
%
%e4.5 #&#
\begin{equation}
\label{tilde-v} \tilde{v}(y,r)\triangleq\mathop{{\inf}}_{\Gamma\in\widetilde
{\mathcal {Y}}(y,r)}\mathbb{E}
\biggl[\int_{0}^{T}\mathit{V}(t,\Gamma
_{t})\,dt \biggr],\qquad (y,r)\in\mathcal{R}.
\end{equation}

The following theorems constitute our main results, and we provide
their proofs through a number of auxiliary results in the next section.

%th4.1 #&#
\begin{theorem}\label{main-1}
Given Assumptions \ref{asssss} and \ref{ass5}, assume also that
conditions (\ref{no-arbitrage}), (\ref{Inada}), (\ref{assAEV}),
(i.e., $\AEE_{0}[\mathit{U}]<\infty$), (\ref{up}) and (\ref{down})
hold true. Moreover, assume that
%
%e4.6 #&#
\begin{equation}
\label{finite-u} \tilde{u}(x,z)<\infty\qquad \mbox{for some }(x,z)\in\mathcal{H},
\end{equation}
then we have:
\begin{longlist}[(ii)]
\item[(i)] The function $\tilde{u}$ is $(-\infty,\infty)$-valued on
$\mathcal{H}$ and $\tilde{v}(y,r)$ is $(-\infty,\infty]$-valued on
$\mathcal{R}$. For each $(y,r)\in\mathcal{R}$ there exists a constant
$s=s(y,r)>0$ such that $\tilde{v}(sy,sr)<\infty$, and the conjugate
duality of value functions $\tilde{u}$ and $\tilde{v}$ holds
\begin{eqnarray*}
\tilde{u}(x,z)&=&\mathop{{\inf}}_{(y,r)\in\mathcal{R}}\bigl\{\tilde {v}(y,r)+xy-zr\bigr
\},\qquad (x,z)\in\mathcal{H},
\\
\tilde{v}(y,r)&=&\mathop{{\sup}}_{(x,z)\in\mathcal{H}}\bigl\{\tilde {u}(x,z)-xy+zr\bigr
\},\qquad (y,r)\in\mathcal{R}.
\end{eqnarray*}
\item[(ii)] The\vspace*{1pt} solution $\Gamma^{\ast}(y,r)$ to the optimization
problem (\ref{tilde-v}) exists and is unique (in the sense of $=$
under $\overline{\mathbb{P}}$ in $\mathbb{L}_{+}^{0}$) for all $(y,r)\in
\mathcal{R}$ such that $\tilde{v}(y,r)<\infty$.
\end{longlist}
\end{theorem}

%
%th4.2 #&#
\begin{theorem}\label{main-2}
In addition to assumptions of Theorem \ref{main-1}, we also assume
that condition (\ref{assAEU}) holds, (i.e., $\AEE_{\infty}[\mathit
{U}]<1$). Then we also have:
\begin{longlist}[(iii)]
\item[(i)] The value function $\tilde{v}(y,r)$ is $(-\infty,\infty
)$-valued on $(y,r)\in\mathcal{R}$, and $\tilde{v}$ is continuously
differentiable on $\mathcal{L}$.
\item[(ii)] The\vspace*{1pt} solution $\tilde{c}^{\ast}(x,z)$ to optimization
problem (\ref{tilde-u}) exists and is unique (in the sense of $=$
under $\overline{\mathbb{P}}$ in $\mathbb{L}_{+}^{0}$) for any $(x,z)\in
\mathcal{H}$, and there exists a representation of the optimal solution
such that $\tilde{c}_{t}^{\ast}(x,z)> 0, \mathbb{P}$-a.s. for all
$t\in[0,T]$.
\item[(iii)] The superdifferential of $\tilde{u}$ maps $\mathcal{H}$
into $\mathcal{R}$, that is,
%
%e4.7 #&#
\begin{equation}
\partial\tilde{u}(x,z)\subset\mathcal{R},\qquad (x,z)\in\mathcal{H}.
\end{equation}
Moreover, if $(y,r)\in\partial\tilde{u}(x,z)$, then there exists a
representation of the optimal solution such that $\Gamma_{t}^{\ast
}(y,r)>0$, $\mathbb{P}$-a.s. for all $t\in[0,T]$ and $\tilde
{c}^{\ast
}(x,z)$ and $\Gamma^{\ast}(y,r)$ are related by
%
%e4.8 #&#
\begin{eqnarray}
\label{optm-soln}  \Gamma_{t}^{\ast}(y,r) &=&\mathit{U}'
\bigl(t,\tilde{c}_{t}^{\ast}(x,z)\bigr)\quad\mbox{or}\nonumber
\\
\tilde{c}_{t}^{\ast}(x,z)&=& \mathit{I}\bigl(t,\Gamma
_{t}^{\ast}(y,r)\bigr),\qquad t\in[0,T],
\\
 {\bigl\langle}\Gamma^{\ast}(y,r),\tilde{c}^{\ast}(x,z) {\bigr
\rangle} &=&xy-zr.
\nonumber
\end{eqnarray}
\item[(iv)] If we restrict the choice of initial wealth $x$ and initial
standard of living $z$ such that $(x,z)\in\overline{\mathcal{H}}\subset
\mathcal{H}$, the solution $c_{t}^{\ast}(x,z)$ to our primal utility
optimization problem (\ref{primalp}) exists and is unique, moreover,
%
%e4.9 #&#
\begin{equation}
\tilde{c}_{t}^{\ast}(x,z)=c_{t}^{\ast}(x,z)-Z
\bigl(c^{\ast}\bigr)_{t}(x,z),\qquad t\in[0,T].
\end{equation}
\end{longlist}
\end{theorem}
%

%s5 #&#
\section{Proofs of main results}\label{section5}

%s5.1 #&#
\subsection{The proof of Theorem \texorpdfstring{\protect\ref{main-1}}{4.1}} The following
proposition will serve as the key step to build Bipolar relationships:

%
%pr5.1 #&#
\begin{proposition}\label{Prop1}
Assume that all conditions of Theorem \ref{main-1} hold true. The
families $ (\widetilde{\mathcal{A}}(x,z) )_{(x,z)\in\mathcal{H}}$
and $ (\widetilde{\mathcal{Y}}(y,r) )_{(y,r)\in\mathcal
{R}}$ have
the following properties:
\begin{longlist}[(ii)]
\item[(i)] For any $(x,z)\in\mathcal{H}$, the set $\widetilde
{\mathcal
{A}}(x,z)$ contains a strictly positive random variable on the product
space. A nonnegative random variable $\tilde{c}$ belongs to
$\widetilde
{\mathcal{A}}(x,z)$ if and only if
%
%e5.1 #&#
\begin{equation}
\label{character-1} {\langle} \tilde{c},\Gamma {\rangle}\leq xy-zr\qquad\mbox{for all }
(y,r)\in\mathcal{R}\mbox{ and }\Gamma\in\widetilde{\mathcal{Y}}(y,r).
\end{equation}

\item[(ii)] For any $(y,r)\in\mathcal{R}$, the set $\widetilde
{\mathcal
{Y}}(y,r)$ contains a strictly positive random variable on the product
space. A nonnegative random variable $\Gamma$ belongs to $\widetilde
{\mathcal{Y}}(y,r)$ if and only if
%
%e5.2 #&#
\begin{equation}
\label{character-2} {\langle} \tilde{c},\Gamma {\rangle}\leq xy-zr\qquad\mbox{for all }
(x,z)\in\mathcal{H}\mbox{ and }\tilde{c}\in\widetilde{\mathcal{A}}(x,z).
\end{equation}
\end{longlist}
\end{proposition}

In order to prove Proposition \ref{Prop1}, for any $p>0$, we denote
by $\mathcal{M}(p)$ the subset of $\mathcal{M}$ that consists of
densities $Y\in\mathcal{M}$ such that $ {\langle} w,Y
{\rangle
}=p$. For any density process $Y\in\mathcal{M}(p)$, define the
auxiliary set as
%
%e5.3 #&#
\begin{eqnarray}
\widetilde{\mathcal{M}}(p)&\triangleq& \biggl\{\Gamma\in\mathbb{L}_{+}^{0}\dvtx
\Gamma_{t}=Y_{t}+\delta_{t}\mathbb{E} \biggl[
\int_{t}^{T}e^{\int
_{t}^{s}(\delta_{v}-\alpha_{v})\,dv}Y_{s}\,ds \Big|
\mathcal{F}_{t} \biggr],
\nonumber\\[-8pt]\\[-8pt]
&&\hspace*{130pt}\forall t\in[0,T], Y\in\mathcal{M}(p) \biggr\}.\nonumber
\end{eqnarray}
We have $ {\langle} \tilde{w},\Gamma {\rangle}=
{\langle
}w,Y {\rangle}=p$.

Define $\mathcal{P}$ as the open interval $\mathcal{P}=(\unde
{p},\bar{p})$, where $\unde{p},\bar{p}$ are given in (\ref
{pbar}) and~(\ref{pdown}). We have the following result.

%
%le5.1 #&#
\begin{lemma}\label{lemma3}
Assume that conditions of Proposition \ref{Prop1} hold true, and let
$p>0$. Then the set $\widetilde{\mathcal{M}}(p)$ is not empty if and
only if $p\in\mathcal{P}=(\unde{p},\bar{p})$, where
$\unde
{p},\bar{p}$ are defined in (\ref{pbar}) and (\ref{pdown}). In
particular,
%
%e5.4 #&#
\begin{equation}
\bigcup_{p\in\mathcal{P}}\widetilde{\mathcal {M}}(p)=\widetilde
{\mathcal{M}},
\end{equation}
where the set $\widetilde{\mathcal{M}}$ is defined by (\ref{settildeM}).
\end{lemma}

\begin{pf}
The proof reduces to verifying that $\mathcal{P}=\mathcal{P}'$, where
\[
\mathcal{P}'\triangleq\bigl\{p>0\dvtx  \widetilde{\mathcal{M}}(p)\neq
\varnothing\bigr\}.
\]

Similar to the proof of Lemma $8$ of Hugonnier and Kramkov \cite
{kram04}, one direction inclusion that $\mathcal{P}\subseteq\mathcal
{P}'$ is obvious.

For the inverse, let $p\in\mathcal{P}'$, $(x,z)\in \operatorname{cl}\mathcal{H}$,
$\Gamma\in\widetilde{\mathcal{M}}(p)$, and we claim that there
exists a
$\tilde{c}\in\widetilde{\mathcal{A}}(x,z)$ such that
\[
\overline{\mathbb{P}}[\tilde{c}>0]>0,
\]
which implies that
\[
0< {\langle} \tilde{c},\Gamma {\rangle}\leq x-zp.
\]
As $(x,z)$ is an arbitrary element of $\operatorname{cl}\mathcal{H}$,
it follows that
$p\in\mathcal{P}$.

As for the above claim, according to Theorem $2.11$ of {Schachermayer} \cite{MR2113724}, condition (\ref{ass5}) guarantees
that for all $Y\in\mathcal{M}$,
\[
\unde{p}<\langle w, Y\rangle<\bar{p},
\]
which is
\[
\unde{p}<\langle\tilde{w}, \Gamma\rangle<\bar{p},
\]
for all the $\Gamma\in\widetilde{\mathcal{M}}$. By the definition of
$\operatorname{cl}\mathcal{H}$ in Lemma \ref{lemKKK}, for any
$(x,z)\in\operatorname{cl}\mathcal
{H}$, we have
\[
x-z {\langle}\tilde{w},\Gamma {\rangle}> 0,
\]
for all $\Gamma\in\widetilde{\mathcal{M}}$, and the claim holds by the
definition of $\widetilde{\mathcal{A}}(x,z)$.
\end{pf}

%
%le5.2 #&#
\begin{lemma}\label{lemma4}
Assume that the conditions of Proposition \ref{Prop1} hold true, let
$p\in\mathcal{P}=(\unde{p},\bar{p})$ and we have $\widetilde
{\mathcal{M}}(p)\subseteq\widetilde{\mathcal{Y}}(1,p)$.
\end{lemma}

\begin{pf}
The conclusion can be directly derived in light of definitions of sets
$\widetilde{\mathcal{A}}(x,z)$ and $\widetilde{\mathcal{Y}}(1,p)$.
\end{pf}

According to the definition of $\widetilde{A}(x,z)$, similar to the
proof of Lemma $10$ of Hugonnier and Kramkov \cite{kram04}, it is
straightforward to show the following result:

%
%le5.3 #&#
\begin{lemma}\label{lemma5}
Assume that conditions of Proposition \ref{Prop1} hold true. For any
$(x,z)\in\mathcal{H}$, a nonnegative random variable $\tilde{c}$
belongs to $\widetilde{\mathcal{A}}(x,z)$ if and only if
%
%e5.5 #&#
\begin{equation}
\label{character-3} {\langle} \tilde{c}, \Gamma {\rangle}\leq x-zp\qquad\mbox{for all
} p\in\mathcal{P}\mbox{ and } \Gamma\in\widetilde{\mathcal{M}}(p).
\end{equation}
\end{lemma}

\begin{pf*}{Proof of Proposition \ref{Prop1}}
For the validity of assertion (i), given any $(x,z)\in\mathcal{H}$,
there exists a $\lambda>0$ such that $(x-\lambda,z)\in\mathcal{H}$
since $\mathcal{H}$ is an open set.

Let $\tilde{c}\in\widetilde{\mathcal{A}}(x-\lambda,z)$, for any
$\Gamma
\in\widetilde{\mathcal{M}}$, and $\tilde{w}_{t}=e^{-\int
_{0}^{t}\alpha
_{v}\,dv}>0$ for all $t\in[0,T]$, we have
%
%e5.6 #&#
\begin{equation}
{\langle}\tilde{c},\Gamma {\rangle}\leq x-\lambda-z {\langle }\tilde{w},\Gamma {
\rangle}.
\end{equation}

By condition (\ref{ass31}) and Proposition \ref{prop31}, we
define $\rho_{t}\triangleq\frac{\lambda}{\bar{p}}\tilde{w}_{t}>0$ for
all $t\in[0,T]$, and then for all $\Gamma\in\widetilde{\mathcal{M}}$,
we obtain
\begin{eqnarray*}
{\langle} \rho, \Gamma {\rangle} & \leq& {\langle} \tilde {c}+\rho, \Gamma {
\rangle} \leq x-\lambda-z {\langle }\tilde {w},\Gamma {\rangle}+\frac{\lambda}{\bar{p}}
{\langle }\tilde {w},\Gamma {\rangle}
\\
& \leq& x-\lambda-z {\langle}\tilde{w},\Gamma {\rangle }+\lambda \leq x-z {
\langle}\tilde{w},\Gamma {\rangle}.
\end{eqnarray*}
Hence, the existence of a strictly positive element $\rho_{t}\in
\widetilde{\mathcal{A}}(x,z)$ follows by the definition of
$\widetilde
{\mathcal{A}}(x,z)$.

If (\ref{character-1}) holds for some $\tilde{c}\in\mathbb
{L}_{+}^{0}$, the density process $\Gamma\in\widetilde{\mathcal{M}}(p)$
belongs to $\widetilde{\mathcal{Y}}(1,p)$ for all $p\in\mathcal{P}$ by
Lemma \ref{lemma4}, and hence (\ref{character-3}) holds. Lemma
\ref{lemma5} then implies that $\tilde{c}\in\widetilde{\mathcal
{A}}(x,z)$. Conversely, suppose $\tilde{c}\in\widetilde{\mathcal
{A}}(x,z)$, the definition of set $\widetilde{\mathcal{Y}}(y,r)$,
$(y,r)\in\mathcal{R}$ implies (\ref{character-1}), and we complete
the proof of assertion (i).

For the proof of the assertion (ii), first we have
\[
k\widetilde{\mathcal{Y}}(y,r)=\widetilde{\mathcal{Y}}(ky,kr)\qquad\mbox{for all }
k>0, (y,r)\in\mathcal{R}.
\]
Therefore, it is enough to consider $(y,r)=(1,p)$ for some $p\in
\mathcal
{P}$. Lemma \ref{lemma4} implies $\Gamma\in\widetilde{\mathcal
{M}}(p)\subseteq\widetilde{\mathcal{Y}}(1,p)$, and the existence of
strictly positive $Y\in\mathcal{M}(p)$ takes care of the existence
$\Gamma\in\widetilde{\mathcal{M}}(p)$ and $\Gamma>0$ $\overline
{\mathbb{P}}$-a.s.

The second part is a direct consequence of the definition of
$\widetilde
{\mathcal{Y}}(y,r)$.
\end{pf*}

For the proof of Theorem \ref{main-1}, we will also need the
following lemmas:

%
%le5.4 #&#
\begin{lemma}\label{lemma6}
Under assumptions of Theorem \ref{main-1}, the value function
$\tilde
{u}$ is $(-\infty,\infty)$-valued on $\mathcal{H}$.
\end{lemma}

\begin{pf}
First, by Lemma \ref{Tech}, the condition $\AEE_{0}[\mathit
{U}]<\infty$
implies that for any positive constant $s>0$, there exist $s_{1}>0$ and
$s_{2}>0$ such that for all $t\in[0,T]$,
%
%e5.7 #&#
\begin{equation}
\mathit{U}(t,x/s) \geq s_{1}\mathit{U}(t,x)+s_{2},\qquad x>0.
\end{equation}

According to condition (\ref{ass41}) and the proof of Proposition
\ref{Prop1}, for each fixed pair $(x,z)\in\mathcal{H}$, there exists
$\lambda=\lambda(x,z)>0$ such that $\frac{\lambda}{\bar{p}}\tilde
{w}_{t}\in\widetilde{\mathcal{A}}(x,z)$, and therefore we deduce that
$\bar{x}\tilde{w}_{t}\in\widetilde{\mathcal{A}}(\frac{\bar
{x}\bar
{p}}{\lambda}x,\frac{\bar{x}\bar{p}}{\lambda}z)$, and
\begin{eqnarray*}
\tilde{u}\biggl(\frac{\bar{x}\bar{p}}{\lambda}x,\frac{\bar{x}\bar
{p}}{\lambda }z\biggr)&=&\sup
_{\tilde{c}\in\widetilde{\mathcal{A}}((({\bar{x}\bar
{p}})/{\lambda})x, (({\bar{x}\bar{p}})/{\lambda})z)}\mathbb{E} \biggl[\int_{0}^{T}
\mathit{U}(t,\tilde{c}_{t})\,dt \biggr]
\\
&\geq&\mathbb{E} \biggl[\int
_{0}^{T}\mathit{U}(t,\bar{x}\tilde{w}_{t})\,dt
\biggr]>-\infty.
\end{eqnarray*}
Hence, for any $(x,z)\in\mathcal{H}$, there exists a constant
$s(x,z)>0$ such that $\tilde{u}(sx,sz)>-\infty$, with $s(x,z)=\frac
{\bar
{x}\bar{p}}{\lambda}$.

For any constant $s>0$,
\[
\widetilde{\mathcal{A}}(x,z)=\widetilde{\mathcal {A}}(sx,sz)/s,
\]
which implies that $\tilde{u}(x,z)>-\infty$ if $\tilde
{u}(sx,sz)>-\infty
$ holds for a constant $s=s(x,z)>0$. By the result above, we can
conclude that $\tilde{u}(x,z)>-\infty$ in the whole domain $\mathcal{H}$.

Now, since the set $\mathcal{H}$ is open, and $\tilde{u}(x,z)<\infty$
for some $(x,z)\in\mathcal{H}$ by the condition (\ref{finite-u}), we
deduce that $\tilde{u}$ is finitely valued on $\mathcal{H}$ by the
concavity of $\tilde{u}$ on~$\mathcal{H}$. And the proof is complete.
\end{pf}

Before we state the next lemma, let us introduce the definition given by
{{\v{Z}}itkovi\'{c}}~\cite{Zit09}.

%de5.1 #&#
\begin{definition}\label{def}
A convex subset $C$ of a topological vector space $X$ is said to be \textit
{convexly compact} if for any nonempty set $A$ and any family $\{F_{a}\}
_{a\in A}$ of closed, convex subsets of $C$, the condition
\[
\forall D\in \operatorname{Fin}(A),\qquad \bigcap_{a\in D}F_{a}
\neq\varnothing\quad\Longrightarrow\quad \bigcap_{a\in A}F_{a}
\neq\varnothing,
\]
where the set $\operatorname{Fin}(A)$ consists of all nonempty finite subsets of $A$
for an arbitrary nonempty set $A$.
\end{definition}

{\v{Z}}itkovi\'{c} \cite{Zit09} furthermore derived an easy
characterization on the space of nonnegative, measurable functions;
see Theorem $3.1$ of {\v{Z}}itkovi\'{c} \cite{Zit09}. We modify his
result to fit into our framework.

%
%pr5.2 #&#
\begin{proposition}\label{compact}
A closed and convex subset $C$ of $\mathbb{L}_{+}^{0}(\Omega\times[0,T],
\mathcal{O},\overline{\mathbb{P}})$ is convexly compact if and only if it is
bounded in the finite measure $\overline{\mathbb{P}}$.
\end{proposition}

Based\vspace*{1pt} on the above proposition, we have the following lemma on the
convexly compactness of sets $\widetilde{\mathcal{A}}(x,z)$ and
$\widetilde{\mathcal{Y}}(y,r)$:

%le5.5 #&#
\begin{lemma}\label{lemma8}
For each pair $(x,z)\in\mathcal{H}$ and $(y,r)\in\mathcal{R}$, the sets
$\widetilde{\mathcal{A}}(x,z)$ and $\widetilde{\mathcal{Y}}(y,r)$ are
convex, solid and closed in the topology of convergence in measure~$\overline{\mathbb{P}}$. Moreover, they are both bounded in $\mathbb
{L}_{+}^{0}(\Omega\times[0,T],\mathcal{O},\overline{\mathbb{P}})$; hence
they are both convexly compact.
\end{lemma}

\begin{pf} For $(y,r)\in\mathcal{R}$, we define auxiliary sets as
%
%e5.8 #&#
\begin{eqnarray}
\label{frakA} \mathfrak{H}(y,r) & \triangleq &\bigl\{(x,z)\in\mathcal{H}\dvtx  xy-zr
\leq 1 \bigr\},
\nonumber
\\[-8pt]
\\[-8pt]
\mathfrak{A}(k) & \triangleq& \bigcup_{(x,z)\in
k\mathfrak {H}(y,r)}
\widetilde{\mathcal{A}}(x,z),
\nonumber
\end{eqnarray}
and denote by $\widetilde{\mathfrak{A}}(k)$ the closure of $\mathfrak
{A}(k)$ with respect to convergence in measure $\overline{\mathbb{P}}$.

From Proposition \ref{Prop1}, it follows that
\[
\Gamma\in\widetilde{\mathcal{Y}}(y,r)\quad\Longleftrightarrow\quad {\langle }\tilde{c},\Gamma
{\rangle}\leq1\qquad \forall\tilde{c}\in \widetilde {\mathfrak{A}}(1).
\]

Hence, sets $\widetilde{\mathcal{Y}}(y,r)$ and $\widetilde{\mathfrak
{A}}(1)$ satisfy
\[
\widetilde{\mathcal{Y}}(y,r) =\widetilde{\mathfrak{A}}(1)^{\circ
}.
\]
At the same time, by its definition, $\widetilde{\mathfrak{A}}(1)$
itself is closed, convex and solid. The bipolar theorem in Brannath and
Schachermayer \cite{Brannath} asserts that $\widetilde{\mathfrak
{A}}(1)=\widetilde{\mathfrak{A}}(1)^{\circ\circ}$, and hence we have
the following Bipolar relationship:
%
%e5.9 #&#
\begin{eqnarray}
\label{eqnbipolar1} \widetilde{\mathfrak{A}}(1) & =& \widetilde{\mathcal{Y}}(y,r)^{\circ},
\nonumber
\\[-8pt]
\\[-8pt]
\widetilde{\mathcal{Y}}(y,r) & =&\widetilde{\mathfrak{A}}(1)^{\circ
}.
\nonumber
\end{eqnarray}

The Bipolar theorem on $\mathbb{L}_{+}^{0}$ implies that $\mathcal
{Y}(y,r)$ is convex, solid and closed under the convergence in measure
$\overline{\mathbb{P}}$.

Similarly, for $(x,z)\in\mathcal{H}$, we define the set
%
%e5.10 #&#
\begin{eqnarray}
\label{frakY} \mathfrak{R}(x,z) & \triangleq&\bigl\{(y,r)\in\mathcal{R}\dvtx  xy-zr
\leq1\bigr\},
\nonumber
\\[-8pt]
\\[-8pt]
\mathfrak{Y}(k) & \triangleq& \bigcup_{(y,r)\in
k\mathfrak {R}(x,z)}
\widetilde{\mathcal{Y}}(y,r),
\nonumber
\end{eqnarray}
and denote by $\widetilde{\mathfrak{Y}}(k)$ the closure of $\mathfrak
{Y}(k)$ with respect to convergence in measure~$\overline{\mathbb{P}}$.

Again, Proposition \ref{Prop1} implies
\[
\tilde{c}\in\widetilde{\mathcal{A}}(x,z)\quad\Longleftrightarrow \quad {\langle} \tilde{c},
\Gamma {\rangle}\leq1\qquad \forall\Gamma\in\widetilde {\mathfrak{Y}},
\]
and the Bipolar relationship
%
%e5.11 #&#
\begin{eqnarray}
\label{eqnbipolar2} \widetilde{\mathfrak{Y}}(1) & =& \widetilde{\mathcal{A}}(x,z)^{\circ},
\nonumber
\\[-8pt]
\\[-8pt]
\widetilde{\mathcal{A}}(x,z) & =&\widetilde{\mathfrak{Y}}(1)^{\circ
}.
\nonumber
\end{eqnarray}

Hence, $\widetilde{\mathcal{A}}(x,z)$ is also convex, solid and closed
under convergence in measure~$\overline{\mathbb{P}}$.

Thanks to the existence of strictly positive $\Gamma\in\widetilde
{\mathcal{M}}(p)$ which is also in $\widetilde{\mathcal{Y}}(1,p)$, the
set $\widetilde{\mathcal{A}}(x,z)$ is therefore bounded in measure
$\overline
{\mathbb{P}}$ by Proposition \ref{Prop1} part (i).

Similarly, as in the proof of Proposition \ref{Prop1}, there exists
$\lambda=\lambda(x,z)$ such that $\rho_t>0$ for all $t\in[0,T]$ and
$\rho_{t}=\frac{\lambda}{\bar{p}}\tilde{w}_{t}\in\widetilde
{\mathcal
{A}}(x,z)$. Due to Proposition \ref{Prop1} part (ii), the set
$\widetilde{\mathcal{Y}}(y,r)$ is also bounded in measure $\overline
{\mathbb
{P}}$. Therefore both of them are convexly compact in $\mathbb{L}_{+}^{0}$.
\end{pf}

Contrary to the existing literature, we cannot mimic the classical
proof of the existence of the dual optimizer due to the lack of
integrability of the dual process $\Gamma\in\widetilde{\mathcal
{Y}}(y,r)$ for $(y,r)\in\mathcal{R}$. Conditions $\AEE_0[\mathit
{U}]<\infty$ and $\mathbb{E}[\int_0^T\mathit{U}(t,\bar{x}\tilde
{w}_t)\,dt]>-\infty$ are critical to prove lemmas below.

%le5.6 #&#
\begin{lemma}\label{lemma9}
Under assumptions of Theorem \ref{main-1}, for each fixed $(y,r)\in
\mathcal{R}$, we have
\[
\mathop{{\sup}}_{\Gamma\in\widetilde{\mathcal{Y}}(y,r)}\mathbb {E} \biggl[\int_{0}^{T}
\mathit{V}^{-}(t,\Gamma_{t})\,dt \biggr]<\infty.
\]
\end{lemma}

\begin{pf}
Condition (\ref{ass41}) admits the existence of $\bar{x}\tilde
{w}_{t}\in\mathbb{L}_{+}^{0}$ such that\break $\mathbb{E} [\int_{0}^{T}\mathit{U}(t,\bar{x}\tilde{w}_{t})\,dt ]>-\infty$. Moreover,
by the proof of Proposition \ref{Prop1}, for each fixed $(y,r)\in
\mathcal{R}$, we can find a pair $(x,z)\in\mathfrak{H}(y,r)$ and there
exists a constant $\lambda(x,z)>0$ such that $\tilde{w}\in\widetilde
{\mathfrak{A}}(\frac{\bar{p}}{\lambda})$, where $\bar{p}$ is
defined by
(\ref{pbar}). Taking into account the inequality $\mathit
{U}(t,x)\leq
\mathit{V}(t,y)+xy$, for any $\Gamma\in\widetilde{\mathcal{Y}}(y,r)$
and $y_{0}(t)\triangleq\inf\{y>0\dvtx  \mathit{V}(t,y)<0\}$, we have
\begin{eqnarray*}
\mathbb{E} \biggl[\int_{0}^{T}
\mathit{V}^{-}(t,\Gamma_{t})\,dt \biggr]&\leq& -\mathbb{E}
\biggl[\int_{0}^{T}\mathit{V}\bigl(t,
\Gamma_{t}\mathbf{1}_{\{
\Gamma_{t}\geq y_{0}(t)\}}+y_{0}(t)
\mathbf{1}_{\{\Gamma_{t}<y_{0}(t)\}
}\bigr)\,dt \biggr]
\\
&\leq& -\mathbb{E} \biggl[\int_{0}^{T}\mathit{U}(t,
\bar{x}\tilde {w}_{t})\,dt \biggr]+\bar{x}\mathbb{E} \biggl[\int
_{0}^{T}\tilde {w}_{t}\Gamma
_{t}\,dt \biggr]
\\
&&{} +\bar{x}\mathbb{E} \biggl[\int_{0}^{T}
\tilde {w}_{t}\bigl(y_{0}(t)-\Gamma_{t}\bigr)
\mathbf{1}_{\{\Gamma_{t}<y_{0}(t)\}
}\,dt \biggr]
\\
&\leq& -\mathbb{E} \biggl[\int_{0}^{T}\mathit{U}(t,
\bar{x}\tilde {w}_{t})\,dt \biggr]+\bar{x}\frac{\bar{p}}{\lambda}+\bar{x}\int
_{0}^{T}y_{0}(t)\,dt.
\end{eqnarray*}
The last term is finitely valued and independent of the initial choice
of $\Gamma$ since $\tilde{w}_{t}\triangleq e^{\int_{0}^{t}(-\alpha
_{v})\,dv}\leq1$ for $t\in[0,T]$ and $\mathop{{\sup }}_{t\in
[0,T]}y_{0}(t)<\infty$ by condition (\ref{down}). Thus the conclusion
holds true.
\end{pf}

%
%le5.7 #&#
\begin{lemma}\label{lemma9-10}
Under assumptions of Theorem \ref{main-1}, for any $(y,r)\in
\mathcal
{R}$,\break $ \{\mathit{V}^{-}(\cdot,\Gamma_{\cdot}) \}_{\Gamma
\in
\widetilde{\mathcal{Y}}(y,r)}$ is uniformly integrable.
\end{lemma}

\begin{pf}
By Lemma \ref{Tech}, the condition $\AEE_{0}[U]<\infty$ is equivalent
to the following assertion:
%
%e5.12 #&#
\begin{equation}
\label{V-1} \exists y_{0}>0\mbox{ and } \mu\in(1,2),\ \forall y\geq
y_{0},\qquad \mathit{V}(t,2y)\geq\mu\mathit{V}(t,y).
\end{equation}

Let $y_{0}>0$ and $\mu\in(1,2)$ be constants in the above inequality
(\ref{V-1}). Taking $\gamma=\log_{2}\mu\in(0,1)$, we define the
auxiliary function $\widetilde{\mathit{V}}(t,y)\dvtx [0,\infty)\times
(0,\infty)\rightarrow\mathbb{R}$ by
%
%e5.13 #&#
\begin{equation}
\widetilde{\mathit{V}}(t,y)\triangleq\cases{ \displaystyle-\frac{2y_{0}}{\gamma}
\mathit{V}'(t,2y_{0})-\mathit {V}(t,y),&\quad$y
\geq2y_{0}$,
\cr
\displaystyle-\mathit{V}(t,2y_{0})-
\frac{2y_{0}}{\gamma}\mathit {V}'(t,2y_{0}) \biggl(
\frac{y}{2y_{0}}\biggr)^{\gamma},&\quad$y<2y_{0}$.}
\end{equation}

For each fixed $t>0$, $\widetilde{\mathit{V}}(t,y)$ is a nonnegative,
concave and nondecreasing function which agrees with $-V(t,y)$
up to a constant for large enough values of $y$ and satisfies
%
%e5.14 #&#
\begin{equation}
\widetilde{\mathit{V}}(t,2y)\leq\mu\widetilde{\mathit {V}}(t,y)\qquad \mbox{for
all } y>0.
\end{equation}

Lemma \ref{lemma9} asserts that
\[
\mathop{{\sup}}_{\Gamma\in\widetilde{\mathcal{Y}}(y,r)}\mathbb {E} \biggl[\int_{0}^{T}
\mathit{V}^{-}(t,\Gamma_{t})\,dt \biggr]<\infty,
\]
and hence in light of the fact that $\mathit{V}^{-}$ and $\widetilde
{\mathit{V}}$ differ only by a constant in a neighborhood of $\infty$,
we will get
%
%e5.15 #&#
\begin{equation}
\label{V-3} \mathop{{\sup}}_{\Gamma\in\widetilde{\mathcal{Y}}(y,r)}\mathbb {E} \biggl[\int
_{0}^{T}\widetilde{\mathit{V}}(t,
\Gamma_{t})\,dt \biggr]<\infty.
\end{equation}

The validity of uniform integrability of the sequence $ (\mathit
{V}^{-}(\cdot,\Gamma_{\cdot}^{n}) )_{n\geq1}$ for $\Gamma
^{n}\in
\widetilde{\mathcal{Y}}(y,r)$, is therefore equivalent to the uniform
integrability of\break $(\widetilde{\mathit{V}}(\cdot,\Gamma_{\cdot
}^{n}))_{n\geq1}$.

To this end, we argue by contradiction. Suppose this sequence is not
uniformly integrable. Then by Rosenthal's subsequence splitting lemma,
we can find a subsequence $(f^{n})_{n\geq1}$, a constant $\varepsilon
>0$ and a disjoint sequence $(A^{n})_{n\geq1}$ of $(\Omega\times
[0,T],\mathcal{O})$ with
\[
A^{n}\in\mathcal{O},\qquad A^{i}\cap A^{j}=\varnothing
\qquad\mbox{if } i\neq j,
\]
such that
\[
\mathbb{E} \biggl[\int_{0}^{T}\widetilde{\mathit
{V}}\bigl(t,f_{t}^{n}\bigr)\mathbf {1}_{A^{n}}\,dt
\biggr]\geq\varepsilon\qquad\mbox{for } n\geq1.
\]

Define the sequence of random variables $(h^{n})_{n\geq1}$
\[
h_{t}^{n}=\sum_{k=1}^{n}f_{t}^{k}
\mathbf{1}_{A^{k}}.
\]

For any $\tilde{c}\in\widetilde{\mathfrak{A}}(1)$, we get
\[
{\bigl\langle} \tilde{c}, h^{n} {\bigr\rangle}\leq\sum
_{k=1}^{n} {\bigl\langle} \tilde{c}, f^{k}
{\bigr\rangle}\leq n.
\]
Hence $\frac{h^{n}}{n}\in\widetilde{\mathcal{Y}}(y,r)$.

One the other hand,
\[
\mathbb{E} \biggl[\int_{0}^{T}\widetilde{\mathit
{V}}\bigl(t,h_{t}^{n}\bigr)\,dt \biggr]\geq\sum
_{k=1}^{n}\mathbb{E} \biggl[\int_{0}^{T}
\widetilde{\mathit {V}}\bigl(t,f_{t}^{k}\bigr)
\mathbf{1}_{A^{k}}\,dt \biggr]\geq\varepsilon n.
\]
Therefore by taking $n=2^{m}$, via iteration, it produces
\begin{eqnarray*}
\mu^{m}\mathop{{\sup}}_{\Gamma_{t}\in\widetilde{\mathcal
{Y}}(y,r)}\mathbb {E} \biggl[\int
_{0}^{T}\widetilde{\mathit{V}}(t,
\Gamma_{t})\,dt \biggr] & \geq& \mu^{m}\mathbb{E} \biggl[\int
_{0}^{T}\widetilde{\mathit{V}}\biggl(t,
\frac
{h_{t}^{2^{m}}}{2^{m}}\biggr)\,dt \biggr]
\\
& \geq&\mathbb{E} \biggl[\int_{0}^{T}\widetilde{
\mathit {V}}\bigl(t,h_{t}^{2^{m}}\bigr)\,dt \biggr]
\geq2^{m}\varepsilon.
\end{eqnarray*}
Since $\mu\in(1,2)$, this contradicts (\ref{V-3}) for m large enough,
and therefore the conclusion holds true.
\end{pf}

Following the classical analysis, Lemma \ref{lemma9-10} together with
Fatou's lemma will deduce the existence of the dual optimizer.

%le5.8 #&#
\begin{lemma}\label{lemma10}
For any pair $(y,r)\in\mathcal{R}$ such that $\tilde{v}(y,r)<\infty$,
the optimal solution $\Gamma^{\ast}$ to the optimization problem
(\ref{tilde-v}) exists and is unique.
\end{lemma}

For the proof of conjugate duality between value functions $\tilde
{u}(x,z)$ and $\tilde{v}(y,r)$, similar to Lemma $11$ of Hugonnier and
Kramkov \cite{kram04}, we have the following result:

%
%le5.9 #&#
\begin{lemma}\label{lemma7}
If $\mathcal{G}\subseteq\mathbb{L}_{+}^{0}$ is convex and contains a
strictly positive random variable, then
\[
\mathop{{\sup}}_{g\in\mathcal{G}}\mathbb{E} \biggl[\int_{0}^{T}
\mathit {U}(t,xg_{t})\,dt \biggr]=\mathop{{\sup}}_{g\in \operatorname{cl}\mathcal{G}}\mathbb
{E} \biggl[\int_{0}^{T}\mathit{U}(t,xg_{t})\,dt
\biggr],\qquad x>0,
\]
where $\operatorname{cl}\mathcal{G}$ denotes the closure of $\mathcal{G}$ with respect
to convergence in measure~$\overline{\mathbb{P}}$.
\end{lemma}

%
%le5.10 #&#
\begin{lemma}\label{lemmanewadd}
For $\tilde{w}_{t}\triangleq e^{\int_{0}^{t}(-\alpha_{v})\,dv}$, we have
the following result:
%
%e5.16 #&#
\begin{equation}
\label{fixresult} \mathbb{E} \biggl[\int_{0}^{T}
\mathit{V}^{-}\bigl(t,\mathit{U}'(t,\tilde
{w}_{t})\bigr)\,dt \biggr]<\infty.
\end{equation}
\end{lemma}

\begin{pf}
Similar to the proof of Lemma \ref{lemma9}, we have $\mathbb
{E}
[\int_{0}^{T}\mathit{U}(t,\bar{x}\tilde{w}_{t})\,dt ]>-\infty$,
and by
the inequality $\mathit{U}(t,x)<\mathit{V}(t,y)+xy$, for any
$y_{0}(t)\triangleq\inf\{y>0\dvtx  \mathit{V}(t,y)<0\}$, we have
%
%e5.17 #&#
\begin{eqnarray}
\label{wogaile3}  &&\mathbb{E} \biggl[\int_{0}^{T}
\mathit{V}^{-}\bigl(t,\mathit{U}'(t,\tilde
{w}_{t})\bigr)\,dt \biggr]
\nonumber
\\
&&\qquad \leq -\mathbb{E} \biggl[\int_{0}^{T}\mathit{V}
\bigl(t,\mathit {U}'(t,\tilde {w}_{t})
\mathbf{1}_{\{\mathit{U}'(t,\tilde{w}_{t})\geq y_{0}(t)\}
}+y_{0}(t)\mathbf{1}_{\{\mathit{U}'(t,\tilde{w}_{t})<y_{0}(t)\}
}\bigr)\,dt
\biggr]\hspace*{-30pt}
\nonumber
\\
&&\qquad \leq -\mathbb{E} \biggl[\int_{0}^{T}\mathit{U}(t,
\bar{x}\tilde {w}_{t})\,dt \biggr]+\bar{x}\mathbb{E} \biggl[\int
_{0}^{T}\tilde {w}_{t}\mathit
{U}'(t,\tilde{w}_{t})\,dt \biggr]
\\
&&\qquad\quad{} +  \bar{x}\mathbb{E} \biggl[\int_{0}^{T}\tilde
{w}_{t}\bigl(y_{0}(t)-\mathit {U}'(t,
\tilde{w}_{t})\bigr)\mathbf{1}_{\{\mathit{U}'(t,\tilde
{w}_{t})<y_{0}(t)\}}\,dt \biggr]
\nonumber
\\
&&\qquad \leq -\mathbb{E} \biggl[\int_{0}^{T}\mathit{U}(t,
\bar{x}\tilde {w}_{t})\,dt \biggr]+\bar{x}\mathbb{E} \biggl[\int
_{0}^{T}\tilde {w}_{t}\mathit
{U}'(t,\tilde{w}_{t})\,dt \biggr]
+\bar{x}\int
_{0}^{T}y_{0}(t)\,dt.
\nonumber
\end{eqnarray}
We already know the first term and the third term are bounded, as for
the second term, we have two different cases:
\begin{longlist}[\textit{Case} 2.]
\item[\textit{Case} 1.] If $\bar{x}\leq1$, the second term can be rewritten as
\begin{eqnarray*}
\mathbb{E} \biggl[\int_{0}^{T}
\tilde{w}_{t}\mathit{U}'(t,\tilde {w}_{t})\,dt
\biggr] & =&\mathbb{E} \biggl[\int_{0}^{T}
\tilde{w}_{t}\mathit {U}'(t,\tilde{w}_{t})
\mathbf{1}_{\{\tilde{w}_{t}\leq x_{0}\}}\,dt \biggr]
\\
&&{} +\mathbb{E} \biggl[\int_{0}^{T}
\tilde{w}_{t}\mathit{U}'(t,\tilde {w}_{t})
\mathbf{1}_{\{\tilde{w}_{t}>x_{0}\}}\,dt \biggr],
\end{eqnarray*}
where $x_{0}$ is the uniform constant in Lemma \ref{Tech} such that
for all $t\in[0,T]$,
%
%e5.18 #&#
\begin{equation}
\label{wogaile} x\mathit{U}'(t,x)< \biggl(\frac{\gamma}{1-\gamma} \biggr)
\bigl(-\mathit {U}(t,x) \bigr)\qquad\mbox{for } 0<x\leq x_{0}.
\end{equation}

Again, by the fact that $\tilde{w}_t\leq1$ for $t\in[0,T]$, it
follows that
\[
\mathbb{E} \biggl[\int_{0}^{T}
\tilde{w}_{t}\mathit{U}'(t,\tilde {w}_{t})
\mathbf{1}_{\{\tilde{w}_{t}>x_{0}\}}\,dt \biggr]<\infty,
\]
and we also have
\begin{eqnarray*}
\mathbb{E} \biggl[\int_{0}^{T}
\tilde{w}_{t}\mathit{U}'(t,\tilde {w}_{t})
\mathbf{1}_{\{\tilde{w}_{t}\leq x_{0}\}}\,dt \biggr]&\leq& - \biggl(\frac{\gamma}{1-\gamma} \biggr)
\mathbb{E} \biggl[\int_{0}^{T}\mathit {U}(t,
\tilde{w}_{t})\mathbf{1}_{\{\tilde{w}_{t}\leq x_{0}\}}\,dt \biggr]
\\
&\leq& \biggl(\frac{\gamma}{1-\gamma} \biggr)\mathbb{E} \biggl[\int
_{0}^{T}\mathit{U}^{-}(t,\bar{x}
\tilde{w}_{t})\,dt \biggr]<\infty,
\end{eqnarray*}
by using inequality (\ref{wogaile}), the increasing property of
$\mathit{U}(t,x)$ with respect to $x$ and condition (\ref{ass41}).

\item[\textit{Case} 2.] If $\bar{x}>1$, the second term can be rewritten as
\begin{eqnarray*}
\mathbb{E} \biggl[\int_{0}^{T}
\tilde{w}_{t}\mathit{U}'(t,\tilde {w}_{t})\,dt
\biggr] & =&\mathbb{E} \biggl[\int_{0}^{T}
\tilde{w}_{t}\mathit {U}'(t,\tilde{w}_{t})
\mathbf{1}_{\{\bar{x}\tilde{w}_{t}\leq x_{0}\}
}\,dt \biggr]
\\
&&{} +\mathbb{E} \biggl[\int_{0}^{T}
\tilde{w}_{t}\mathit{U}'(t,\tilde {w}_{t})
\mathbf{1}_{\{\bar{x}\tilde{w}_{t}>x_{0}\}}\,dt \biggr],
\end{eqnarray*}
where $x_{0}$ is the uniform constant in Lemma \ref{Tech} such that
for all $t\in[0,T]$, inequality (\ref{wogaile}) holds and moreover,
%
%e5.19 #&#
\begin{equation}
\label{wogaile2} \mathit{U}\biggl(t,\frac{1}{\bar{x}} x\biggr)>\biggl(
\frac{1}{\bar{x}}\biggr)^{-{\gamma}/{(1-\gamma)}}\mathit{U}(t,x)\qquad\mbox{for } 0<x\leq
x_{0},
\end{equation}
holds for all $t\in[0,T]$.

Again, the second term is bounded since $\bar{x}\tilde{w}_t\leq\bar{x}$
for $t\in[0,T]$, and for the first term, we have
\begin{eqnarray*}
\mathbb{E} \biggl[\int_{0}^{T}
\tilde{w}_{t}\mathit{U}'(t,\tilde {w}_{t})
\mathbf{1}_{\{\bar{x}\tilde{w}_{t}\leq x_{0}\}}\,dt \biggr] & \leq& - \biggl(\frac{\gamma}{1-\gamma} \biggr)
\mathbb{E} \biggl[\int_{0}^{T}\mathit {U}(t,
\tilde{w}_{t})\mathbf{1}_{\{\bar{x}\tilde{w}_{t}\leq x_{0}\}
}\,dt \biggr]
\\
& \leq& \biggl(\frac{\gamma}{1-\gamma} \biggr) \biggl(\frac{1}{\bar
{x}}
\biggr)^{-{\gamma}/(1-\gamma)}\mathbb{E} \biggl[\int_{0}^{T}
\mathit {U}^{-}(t,\bar {x}\tilde{w}_{t})\,dt \biggr]
\\
&<& \infty
\end{eqnarray*}
\end{longlist}
by inequalities (\ref{wogaile}) and (\ref{wogaile2}) and condition
(\ref{ass41}).

Hence the second term in (\ref{wogaile3}) is finite, and therefore
result (\ref{fixresult}) holds true.
\end{pf}

We emphasize that we have to revise the classical Minimax theorem based
on~$\mathbb{L}^{1}$ to derive the important conjugate duality
relationship. The\vspace*{1pt} following Minimax theorem by Kauppila \cite{Helena}
can serve as a substitute tool on the space $\mathbb{L}_{+}^{0}$ for
convexly compact sets.

%
%th5.1 #&#
\begin{theorem}[(Minimax theorem)]\label{Minimax}
Let $A$ be a nonempty convex subset of a topological space, and $B$ a
nonempty, closed, convex and convexly compact subset of a topological
vector space. Let $H\dvtx  A\times B\rightarrow\mathbb{R}$ be convex on $A$,
and concave and upper-semicontinuous on $B$. Then
\[
\sup_{B}\inf_{A}H=\inf
_{A}\sup_{B}H.
\]
\end{theorem}

See Theorem A.1 in Appendix A of Kauppila \cite{Helena}.

%
%le5.11 #&#
\begin{lemma}\label{lemma11}
Under assumptions of Theorem \ref{main-1}, the conjugate duality
results hold
%
%e5.20 #&#
\begin{eqnarray}
\label{conjugate}  \tilde{u}(x,z)&=&\mathop{{\inf}}_{(y,r)\in\mathcal{R}}\bigl\{\tilde
{v}(y,r)+xy-zr\bigr\},\qquad (x,z)\in\mathcal{H},
\nonumber
\\[-8pt]
\\[-8pt]
 \tilde{v}(y,r)&=&\mathop{{\sup}}_{(x,z)\in\mathcal{H}}\bigl\{\tilde {u}(x,z)-xy+zr\bigr
\},\qquad (y,r)\in\mathcal{R}.
\nonumber
\end{eqnarray}
\end{lemma}

\begin{pf}
For $n>0$, we define $\mathcal{S}_{n}$ as a subset in $\mathbb
{L}_{+}^{0}(\Omega\times[0,T],\mathcal{O},\overline{\mathbb{P}})$ by
\[
\mathcal{S}_{n}=\bigl\{\tilde{c}\in\mathbb{L}_{+}^{0}\dvtx 0
\leq\tilde {c}\leq n\tilde{w}\bigr\}.
\]
It is clear that sets $\mathcal{S}_{n}$ are closed, convex and bounded
in probability and hence convexly compact in $\mathbb{L}_{+}^{0}$.

It is easy to verify that the functional
\[
\tilde{c}\mapsto\mathbb{E} \biggl[\int_{0}^{T}
\bigl(\mathit {U}(t,\tilde {c}_{t})-\tilde{c}_{t}
\Gamma_{t} \bigr)\,dt \biggr]
\]
is upper-semicontinuous on $\mathcal{S}_{n}$ under convergence in
measure $\overline{\mathbb{P}}$, for all $\Gamma\in\widetilde{\mathcal
{Y}}(y,r)$ and $(y,r)\in\mathcal{R}$.

Lemma \ref{lemma8} implies that $\widetilde{\mathcal{Y}}(y,r)$ is a
closed convex subset of $\mathbb{L}_{+}^{0}$. We can use the above
Minimax Theorem \ref{Minimax} to get the following equality: for
fixed~$n$,
\begin{eqnarray*}
&& \mathop{{\sup}}_{\tilde{c}\in\mathcal{S}_{n}}\mathop{{\inf }}_{\Gamma\in \widetilde{\mathcal{Y}}(y,r)}\mathbb{E}
\biggl[\int_{0}^{T} \bigl(\mathit{U}(t,
\tilde{c}_{t})-\tilde{c}_{t}\Gamma_{t} \bigr)\,dt
\biggr]
\\
&&\qquad =\mathop{{\inf}}_{\Gamma\in\widetilde{\mathcal{Y}}(y,r)}\mathop {{\sup}}_{\tilde{c}\in\mathcal{S}_{n}}\mathbb{E}
\biggl[\int_{0}^{T} \bigl(\mathit{U}(t,
\tilde{c}_{t})-\tilde{c}_{t}\Gamma_{t} \bigr)\,dt
\biggr].
\end{eqnarray*}

By the bipolar relationship (\ref{eqnbipolar1}) and the definition,
we have
%
%e5.21 #&#
\begin{equation}
\bigcup_{(x,z)\in\mathcal{H}}\widetilde{\mathcal{A}}(x,z)=\bigcup
_{k>0}\widetilde{\mathfrak{A}}(k).
\end{equation}

To continue the proof, we define the auxiliary set
\[
\mathfrak{A}'(k)\triangleq \Bigl\{\tilde{c}\in\widetilde{\mathfrak
{A}}(k)\dvtx  \sup_{\Gamma\in\widetilde{\mathcal{Y}}(y,r)}\langle \tilde{c}, \Gamma\rangle=k \Bigr\},
\]
and clearly, it follows that
%
%e5.22 #&#
\begin{equation}
\label{conj-1} \bigcup_{k>0}\widetilde{
\mathfrak{A}}(k)=\bigcup_{(x,z)\in\mathcal
{H}}\widetilde{
\mathcal{A}}(x,z)=\bigcup_{k>0}\mathfrak{A}'(k).
\end{equation}
We first show that
%
%e5.23 #&#
\begin{eqnarray}
\label{conj-2}
&& \mathop{{\lim}}_{n\rightarrow\infty}\mathop{{\sup}}_{\tilde
{c}\in\mathcal {S}_{n}}
\mathop{{\inf }}_{\Gamma\in\widetilde
{\mathcal{Y}}(y,r)}\mathbb{E} \biggl[\int_{0}^{T}
\bigl(\mathit {U}(t,\tilde{c}_{t})-\tilde {c}_{t}
\Gamma_{t} \bigr)\,dt \biggr]
\nonumber
\\[-8pt]
\\[-8pt]
&&\qquad  =\mathop{{\sup}}_{k>0}\mathop{{\sup }}_{\tilde{c}\in\mathfrak
{A}'(k)}\mathop{{
\inf}}_{\Gamma\in\widetilde{\mathcal
{Y}}(y,r)}\mathbb{E} \biggl[\int_{0}^{T}
\bigl(\mathit{U}(t,\tilde{c}_{t})-\tilde{c}_{t}\Gamma
_{t} \bigr)\,dt \biggr].
\nonumber
\end{eqnarray}

The direction of inequality \textquotedblleft$\geq$''
holds by
\begin{eqnarray*}
&& \mathop{{\lim}}_{n\rightarrow\infty}\mathop{{\sup}}_{\tilde
{c}\in\mathcal {S}_{n}}\mathop{{\inf
}}_{\Gamma\in\widetilde
{\mathcal{Y}}(y,r)}\mathbb{E} \biggl[\int_{0}^{T}
\bigl(\mathit {U}(t,\tilde{c}_{t})-\tilde {c}_{t}
\Gamma_{t} \bigr)\,dt \biggr]
\\
&&\qquad \geq \mathop{{\lim }}_{n\rightarrow\infty }\mathop{{\sup }}_{\tilde{c}\in\mathfrak
{A}'(k)\cap\mathcal{S}_{n}}\mathop{{
\inf}}_{\Gamma\in\widetilde
{\mathcal{Y}}(y,r)}\mathbb{E} \biggl[\int_{0}^{T}
\bigl(\mathit{U}(t,\tilde{c}_{t})-\tilde{c}_{t}\Gamma
_{t} \bigr)\,dt \biggr]
\\
&&\qquad =  \mathop{{\sup}}_{\tilde{c}\in\mathfrak{A}'(k)}\mathop{{\inf }}_{\Gamma\in \widetilde{\mathcal{Y}}(y,r)}\mathbb{E}
\biggl[\int_{0}^{T} \bigl(\mathit{U}(t,
\tilde{c}_{t})-\tilde{c}_{t}\Gamma_{t} \bigr)\,dt
\biggr]\qquad \forall k>0.
\end{eqnarray*}
The other direction \textquotedblleft$\leq$'' is
obvious since for any $(x,z)\in\mathcal{H}$, we have $n\tilde{w}\in
\mathfrak{A}'(n\bar{p})$, and hence $\mathcal{S}_{n}\subset
\mathfrak
{A}'(n\bar{p})$.

To continue the proof, we need to prepare finiteness results as below.

From definitions in Lemma \ref{lemma8} and by Lemma \ref{lemma7},
it is easy to see that
%
%e5.24 #&#
\begin{eqnarray}\label{finite}
\mathop{{\sup}}_{\tilde{c}\in\widetilde{ \mathfrak{A}}(k)}\mathbb {E} \biggl[\int
_{0}^{T}\mathit{U}(t,\tilde{c}_{t})\,dt
\biggr] &=& \mathop{{\sup }}_{\tilde{c}\in \mathfrak{A}(k)}\mathbb{E} \biggl[\int
_{0}^{T}\mathit{U}(t,\tilde {c}_{t})\,dt
\biggr]
\nonumber\\[-8pt]\\[-8pt]
&=& \mathop{{\sup}}_{(x,z)\in k\mathfrak{H}(y,r)}\tilde {u}(x,z),\qquad k>0,\nonumber
\end{eqnarray}
and we claim that
%
%e5.25 #&#
\begin{equation}
\label{frak-u} \mathop{{\sup}}_{(x,z)\in k\mathfrak{H}(y,r)}\tilde{u}(x,z)<\infty,\qquad k>0.
\end{equation}

Since the set $\mathcal{R}$ is open, and the set $\mathfrak{H}(y,r)$ is
bounded, (\ref{frak-u}) follows from the concavity of $\tilde{u}$ and
$\tilde{u}(x,z)<\infty$ for all $(x,z)\in\mathcal{H}$.

Now, by (\ref{conj-1}), (\ref{frak-u}) and the definition of domain
$\mathcal{H}$, we have further equalities:
\begin{eqnarray*}
&& \mathop{{\sup}}_{k>0}\mathop{{\sup }}_{\tilde{c}\in\mathfrak
{A}'(k)}\mathop{{
\inf}}_{\Gamma\in\widetilde{\mathcal
{Y}}(y,r)}\mathbb{E} \biggl[\int_{0}^{T}
\bigl(\mathit{U}(t,\tilde{c}_{t})-\tilde{c}_{t}\Gamma
_{t} \bigr)\,dt \biggr]
\\
&&\qquad =  \mathop{{\sup}}_{k>0} \biggl\{\mathop{{\sup}}_{\tilde{c}\in
\mathfrak {A}'(k)}
\mathbb{E} \biggl[\int_{0}^{T}\mathit{U}(t,\tilde
{c}_{t})\,dt \biggr]-k \biggr\}
\\
&&\qquad =\mathop{{\sup}}_{k>0} \biggl\{
\mathop{{\sup }}_{\tilde {c}\in\widetilde{\mathfrak{A}}(k)}\mathbb{E} \biggl[\int_{0}^{T}
\mathit{U}(t,\tilde{c}_{t})\,dt \biggr]-k \biggr\}
\\
&&\qquad =  \mathop{{\sup}}_{k>0} \bigl\{\mathop{{\sup}}_{(x,z)\in
k\mathfrak {H}(y,r)}
\tilde{u}(x,z)-k \bigr\}=\mathop{{\sup }}_{(x,z)\in\mathcal {H}}\bigl\{
\tilde{u}(x,z)-xy+zr\bigr\}.
\end{eqnarray*}

On the other hand,
\begin{eqnarray*}
&& \mathop{{\inf}}_{\Gamma\in\widetilde{\mathcal{Y}}(y,r)}\mathop {{\sup}}_{\tilde {c}\in\mathcal{S}_{n}}\mathbb{E} \biggl[
\int_{0}^{T} \bigl(\mathit {U}(t,
\tilde{c}_{t})-\tilde{c}_{t}\Gamma_{t} \bigr)\,dt
\biggr]
\\
&&\qquad =\mathop {{\inf}}_{\Gamma\in\widetilde{\mathcal{Y}}(y,r)}\mathbb{E} \biggl[\int_{0}^{T}
\mathit{V}^{n}(t,\Gamma_{t},\omega)\,dt \biggr]\triangleq
\tilde {v}^{n}(y,r),
\end{eqnarray*}
where we define $\mathit{V}^{n}(t,y,\omega)$ according to the
definition of set $\mathcal{S}_{n}$ as
\[
\mathit{V}^{n}(t,y,\omega)=\mathop{{\sup}}_{0<x\leq n\tilde
{w}} \bigl[
\mathit{U}(t,x)-xy \bigr].
\]

Consequently, it is sufficient to show that
\[
\mathop{{\lim}}_{n\rightarrow\infty}\tilde{v}^{n}(y,r)=\mathop {{
\lim}}_{n\rightarrow\infty}\mathop{{\inf}}_{\Gamma\in
\widetilde{\mathcal {Y}}(y,r)}\mathbb{E} \biggl[\int
_{0}^{T}\mathit {V}^{n}(t,\Gamma
_{t},\omega)\,dt \biggr]=\tilde{v}(y,r),\qquad (y,r)\in\mathcal{R}.
\]

Evidently, $\tilde{v}^{n}(y,r)\leq\tilde{v}(y,r)$, for $n\geq1$. Let
$(\Gamma^{n})_{n\geq1}$ be a sequence in $\widetilde{\mathcal
{Y}}(y,r)$ such that
\[
\mathop{{\lim}}_{n\rightarrow\infty}\mathbb{E} \biggl[\int_{0}^{T}
\mathit {V}^{n}\bigl(t,\Gamma_{t}^{n},\omega
\bigr)\,dt \biggr]=\mathop{{\lim }}_{n\rightarrow\infty }\tilde{v}^{n}(y,r).
\]

There exists a sequence $h^{n}\in\operatorname{conv}(\Gamma
^{n},\Gamma
^{n+1},\ldots)$, $n\geq1$, converging almost surely to a random
variable $\Gamma$. Then $\Gamma\in\widetilde{\mathcal{Y}}(y,r)$ is
verified because the set $\widetilde{\mathcal{Y}}(y,r)$ is closed under
convergence in finite measure $\overline{\mathbb{P}}$.

We claim that the sequence of processes $(\mathit{V}^{n}(\cdot,
h_{\cdot
}^{n},\omega)^{-}), n\geq1$ is uniformly integrable. In fact, we can rewrite
%
%e5.26 #&#
\begin{eqnarray}\label{change1}
\bigl(\mathit{V}^{n}\bigl(t,h_{t}^{n},
\omega\bigr) \bigr)^{-} &=& \bigl(\mathit {V}^{n}
\bigl(t,h_{t}^{n},\omega\bigr) \bigr)^{-}
\mathbf{1}_{\{h_{t}^{n}\leq
\mathit
{U}'(t,\tilde{w}_{t})\}}
\nonumber\\[-8pt]\\[-8pt]
&&{} + \bigl(\mathit{V}^{n}\bigl(t,h_{t}^{n},
\omega \bigr) \bigr)^{-}\mathbf{1}_{\{h_{t}^{n}>\mathit{U}'(t,\tilde{w}_{t})\}},\nonumber
\end{eqnarray}
since $\mathit{V}^{n}(t,y,\omega)=\mathit{V}(t,y)$ for $y\geq
\mathit
{U}'(t,\tilde{w}_{t})\geq\mathit{U}'(t,n\tilde{w}_{t})$ by definition.
The proof of Lemma \ref{lemma9-10} gives the uniform integrability of
the sequence of processes $ (\mathit{V}^{n}(\cdot,h_{\cdot
}^{n},\omega) )^{-}\mathbf{1}_{\{h_{\cdot}^{n}>\mathit
{U}'(\cdot,\tilde{w}_{\cdot})\}}, n\geq1$.

On the other hand, by the monotonicity of $(V^{n})^{-}$, for all $n>1$,
%
%e5.27 #&#
\begin{eqnarray}
\bigl(\mathit{V}^{n}\bigl(t,h_{t}^{n},\omega
\bigr) \bigr)^{-}\mathbf{1}_{\{
h_{t}^{n}\leq \mathit{U}'(t,\tilde{w}_{t})\}}&\leq& \bigl(\mathit
{V}^{1}\bigl(t,h_{t}^{n},\omega\bigr)
\bigr)^{-}\mathbf{1}_{\{h_{t}^{n}\leq
\mathit
{U}'(t,\tilde{w}_{t})\}}
\nonumber\\[-8pt]\\[-8pt]
&\leq& \bigl(\mathit{V}\bigl(t,
\mathit{U}'(t,\tilde {w}_{t})\bigr) \bigr)^{-}.\nonumber
\end{eqnarray}
By Lemma \ref{lemmanewadd}, the right-hand side is integrable in the
product space, and hence the sequence $ (\mathit{V}^{n}(\cdot,h_{\cdot}^{n},\omega) )^{-}\mathbf{1}_{\{h_{\cdot}^{n}\leq
\mathit
{U}'(\cdot,\tilde{w}_{\cdot})\}}, n\geq1$ is also uniformly
integrable. Thus our claim holds true. Moreover, we have the following
inequalities:
\begin{eqnarray*}
\mathop{{\lim}}_{n\rightarrow\infty}\mathbb{E} \biggl[\int_{0}^{T}
\mathit {V}^{n}\bigl(t,\Gamma_{t}^{n},\omega
\bigr)\,dt \biggr]&\geq& \mathop{{\lim\inf }}_{n\rightarrow \infty}\mathbb{E} \biggl[\int
_{0}^{T}\mathit {V}^{n}
\bigl(t,h_{t}^{n},\omega\bigr)\,dt \biggr]
\\
&\geq& \mathbb{E} \biggl[\int_{0}^{T}\mathit{V}(t,
\Gamma_{t})\,dt \biggr]\geq \tilde{v}(y,r),
\end{eqnarray*}
which prove
%
%e5.28 #&#
\begin{equation}
\label{tildevtildef} \tilde{v}(y,r)=\sup_{(x,z)\in\mathcal{H}}\bigl\{
\tilde{u}(x,z)-xy+zr\bigr\}.
\end{equation}

Equality (\ref{conjugate}) is a direct consequence of equality
(\ref{tildevtildef}) and the properties of convex conjugation; see
Corollary $12.2.2$ and Theorem $12.2$ in Rockafellar \cite{Rock}.
\end{pf}

\begin{pf*}{Proof of Theorem \ref{main-1}}
It is now sufficient to show that the conjugate value function $\tilde
{v}$ is $(-\infty,\infty]$-valued on $\mathcal{R}$.

The Legendre--Fenchel transform gives that
\[
\mathit{U}(t,x)\leq\mathit{V}(t,y)+xy.
\]
By integration, it is easy to see for any $\tilde{c}\in\widetilde
{\mathcal{A}}(x,z)$ and $\Gamma\in\widetilde{\mathcal{Y}}(y,r)$,
\[
\mathbb{E} \biggl[\int_{0}^{T}\mathit{U}(t,
\tilde{c}_{t})\,dt \biggr]\leq \mathbb {E} \biggl[\int
_{0}^{T}\mathit{V}(t,\Gamma_{t})\,dt
\biggr]+\mathbb {E} \biggl[\int_{0}^{T}
\tilde{c}_{t}\Gamma_{t}\,dt \biggr].
\]
Proposition \ref{Prop1} deduces that
\[
\tilde{u}(x,z)\leq\tilde{v}(y,r)+xy-zr.
\]
Hence for all $(y,r)\in\mathcal{R}$, we have $\tilde{v}(y,r)>-\infty$
by Lemma \ref{lemma6}.

On the other hand, thanks to conjugate duality (\ref{conjugate}) and
Bipolar relationship~(\ref{eqnbipolar1}), we can follow proofs of
Lemmas \ref{lemma8}~and~\ref{lemma11} and obtain that for each
fixed $(y,r)\in\mathcal{R}$,
\[
\mathop{{\sup}}_{(x,z)\in k\mathfrak{H}(y,r)}\tilde{u}(x,z)=\mathop {{\inf}}_{s>0}
\bigl\{\tilde{v}(sy,sr)+ks\bigr\}.
\]
The finiteness result (\ref{frak-u}) for all $k>0$ in the proof of
Lemma \ref{lemma11} guarantees the existence of a constant $s(y,r)>0$
such that $\tilde{v}(sy,sr)<\infty$.
\end{pf*}

%s5.2 #&#
\subsection{The proof of Theorem \texorpdfstring{\protect\ref{main-2}}{4.2}}
Let us move on to the proof of Theorem \ref{main-2} and some further
lemmas and auxiliary results are needed.

%
%le5.12 #&#
\begin{lemma}\label{lemma12}
Under assumptions of Theorem \ref{main-2}, we have $\tilde{v}(y,r)$
is $(-\infty,\infty)$-valued on $\mathcal{R}$.
\end{lemma}

\begin{pf}
Similar to the proof of Lemma \ref{lemma6}, under the additional
condition (\ref{assAEU}), we can show that $\tilde{v}(y,r)<\infty$
if $\tilde{v}(sy,sr)<\infty$ for a constant $s=s(y,r)>0$. However, it
has been shown that Theorem \ref{main-1} admits the existence of
$s=s(y,r)>0$.
\end{pf}

Notice that we cannot mimic proofs of Lemmas \ref{lemma9}, \ref{lemma9-10} and \ref{lemma10} to obtain the existence of the optimal
solution to the problem (\ref{tilde-u}). In fact, our arguments for
the dual problem depend on the existence of a bounded process $\tilde
{w}\in\widetilde{\mathfrak{A}}(\frac{\bar{p}}{\lambda})$, which is
missing in the dual space. To this end, we resort to another auxiliary
optimization problem and take advantage of the Bipolar results built in
Lemma \ref{lemma8}.

%le5.13 #&#
\begin{lemma}\label{lemma13}
Define the auxiliary optimization problem to the auxiliary dual utility
minimization problem (\ref{tilde-v}) as
%
%e5.29 #&#
\begin{equation}
\label{hat-v} \hat{v}(k)=\inf_{\Gamma\in\widetilde{\mathfrak{Y}}(k)}\mathbb {E} \biggl[\int
_{0}^{T}\mathit{V}(t,\Gamma_{t})\,dt
\biggr],
\end{equation}
where $\widetilde{\mathfrak{Y}}(k)$ is defined in Lemma \ref{lemma8}
as the bipolar set of $\widetilde{\mathcal{A}}(x,z)$ on the product
space for any $(x,z)\in\mathcal{H}$.

Then, for all $k>0$, under hypothesis of Theorem \ref{main-2}, the
value function \mbox{$\hat{v}(k)<\infty$} for all $k>0$, and
the optimal solution $\widehat{\Gamma}(k)$ exists and is unique and $\widehat{\Gamma}_{t}(k)>0$ for all $t\in[0,T]$. Moreover, for each $k>0$, and
any $\Gamma\in\widetilde{\mathfrak{Y}}(k)$, we have
\[
\mathbb{E} \biggl[\int_{0}^{T}\bigl(
\Gamma_{t}-\widehat{\Gamma}_{t}(k)\bigr)\mathit {I}\bigl(t,
\widehat{\Gamma}_{t}(k)\bigr)\,dt \biggr]\leq0.
\]
\end{lemma}

\begin{pf}
According to the definition in Lemma \ref{lemma8} and by Lemma \ref{lemma12}, it is easy to see
\begin{eqnarray*}
\hat{v}(k)&=& \mathop{{\inf}}_{\Gamma\in\widetilde{\mathfrak
{Y}}(k)}\mathbb {E} \biggl[\int
_{0}^{T}\mathit{V}(t,\Gamma_{t})\,dt
\biggr]\leq \mathop {{\inf}}_{\Gamma\in\mathfrak{Y}(k)}\mathbb{E} \biggl[\int
_{0}^{T}\mathit {V}(t,\Gamma_{t})\,dt
\biggr]
\\
&= & \mathop{{\inf}}_{(y,r)\in k\mathfrak{R}(x,z)}\tilde {v}(y,r)<\infty,\qquad k>0.
\end{eqnarray*}

Taking into account the Bipolar relationship (\ref{eqnbipolar2}), we
have that $\widetilde{\mathfrak{Y}}(k)$ is convexly compact in
$\mathbb
{L}_{+}^{0}$, and the existence and uniqueness of optimal solution
$\widehat{\Gamma}(k)$ will follow the similar proof of Theorem \ref{main-1}.

For\vspace*{2pt} $k>0$, $\epsilon\in(0,1)$, we define $\Gamma_{t}^{\epsilon
}=(1-\epsilon)\widehat{\Gamma}_{t}(k)+\epsilon\Gamma_{t}$, and for all
$t\in
[0,T]$, the optimality of $\widehat{\Gamma}(k)$ implies
%
%e5.30 #&#
\begin{eqnarray}
\label{wu} 0&\leq& \frac{1}{\epsilon}\mathbb{E} \biggl[\int
_{0}^{T} \bigl(\mathit {V}\bigl(t,
\Gamma_{t}^{\epsilon}\bigr)-\mathit{V}\bigl(t,\widehat{\Gamma}_{t}(k)\bigr) \bigr)\,dt \biggr]
\nonumber
\\
&\leq& \frac{1}{\epsilon}\mathbb{E} \biggl[\int_{0}^{T}
\bigl(\widehat{\Gamma }_{t}(k)-\Gamma_{t}^{\epsilon}
\bigr)\mathit{I}\bigl(t,\Gamma _{t}^{\epsilon
}\bigr)\,dt \biggr]
\\
&=& \mathbb{E} \biggl[\int_{0}^{T} \bigl(\widehat{\Gamma}_{t}(k)-\Gamma _{t} \bigr)\mathit{I}\bigl(t,
\Gamma_{t}^{\epsilon}\bigr)\,dt \biggr].\nonumber
\end{eqnarray}

We claim that the family $ \{ ((\Gamma_{t}-\widehat{\Gamma}_{t}(k))\mathit{I}(t,\Gamma^{\epsilon}_{t}) )^{-}, \epsilon\in
(0,1) \}$ is uniformly integrable with respect to $\overline{\mathbb
{P}}$. Observe that
\[
\bigl(\bigl(\Gamma_{t}-\widehat{\Gamma}_{t}(k)\bigr)\mathit{I}
\bigl(t,\Gamma ^{\epsilon
}_{t}\bigr) \bigr)^{-}\leq\widehat{\Gamma}_{t}(k)\mathit{I}\bigl(t,\Gamma _{t}^{\epsilon
}
\bigr)\leq\widehat{\Gamma}_{t}(k)\mathit{I}\bigl(t,(1-\epsilon)\widehat{\Gamma}_{t}(k)\bigr)\qquad \forall t\in[0,T].
\]

For fixed $\epsilon_{0}<1$ and $\epsilon<\epsilon_{0}$, we deduce that
for each $t\in[0,T]$,
\begin{eqnarray*}
 \bigl|\widehat{\Gamma}_{t}(k)\mathit{I}\bigl(t,(1-\epsilon)\widehat{\Gamma}_{t}(k)\bigr) \bigr|&\leq& \bigl|\widehat{\Gamma}_{t}(k)\mathit{I}
\bigl(t,(1-\epsilon)\widehat{\Gamma}_{t}(k)\bigr) \bigr|\mathbf{1}_{\{\widehat{\Gamma}_{t}(k)\leq y_{1}\}}
\\
&&{} +  \bigl|\widehat{\Gamma}_{t}(k)\mathit{I}\bigl(t,(1-\epsilon)\widehat{\Gamma}_{t}(k)\bigr) \bigr|\mathbf{1}_{\{\widehat{\Gamma}_{t}(k)\geq
{y_{2}}/({1-\epsilon_{0}})\}}
\\
&&{} + \bigl|\widehat{\Gamma}_{t}(k)\mathit{I}\bigl(t,(1-\epsilon)\widehat{\Gamma}_{t}(k)\bigr) \bigr|\mathbf{1}_{\{y_{1}<\widehat{\Gamma}_{t}(k)<
{y_{2}}/({1-\epsilon_{0}})\}}.
\end{eqnarray*}

By Lemma \ref{Tech}, reasonable asymptotic elasticity conditions
$\AEE_{0}[\mathit{U}]<\infty$ and $\AEE_{\infty}[\mathit{U}]<1$ imply that
for fixed $\mu>0$, there exist constants $C_{1}>0$, $C_{2}>0$,
$y_{1}>0$ and $y_{2}>0$ such that
%
%e5.31 #&#
\begin{eqnarray}
\label{VVV} -\mathit{V}'(t,\mu y) & <&C_{1}
\frac{\mathit{V}(t,y)}{y}\qquad\mbox{for } 0<y\leq y_{1},
\nonumber
\\[-8pt]
\\[-8pt]
-\mathit{V}'(t,y) & <&C_{2}\frac{-\mathit{V}(t,y)}{y}
\qquad\mbox{for } y_{2}\leq y.
\nonumber
\end{eqnarray}
Hence, the first term is dominated by
\[
\bigl|\widehat{\Gamma}_{t}(k)\mathit{I}\bigl(t,(1-\epsilon)\widehat{\Gamma}_{t}(k)\bigr) \bigr|\mathbf{1}_{\{\widehat{\Gamma}_{t}(k)\leq y_{1}\}}\leq \frac
{1}{1-\epsilon_{0}}C_{1}\bigl|
\mathit{V}\bigl(t,\widehat{\Gamma}_{t}(k)\bigr)\bigr|,
\]
and the second term is dominated by
\begin{eqnarray*}
&& \bigl|\widehat{\Gamma}_{t}(k)\mathit{I}\bigl(t,(1-\epsilon)\widehat{\Gamma}_{t}(k)\bigr) \bigr|\mathbf{1}_{\{\widehat{\Gamma}_{t}(k)
\geq {y_{2}}/({1-\epsilon_{0}})\}}
\\
&&\qquad \leq \frac{-1}{1-\epsilon
_{0}}C_{2}\mathit {V}\bigl(t,(1-\epsilon)\widehat{\Gamma}_{t}(k)\bigr)
\mathbf{1}_{\{\widehat{\Gamma}_{t}(k)\geq {y_2}/{(1-\epsilon_0)}\}}
\\
&&\qquad \leq \frac{1}{1-\epsilon_{0}}C_{2}\bigl|\mathit{V}\bigl(t,\widehat{\Gamma}_{t}(k)\bigr)\bigr|.
\end{eqnarray*}

These two terms are both in $\mathbb{L}^{1}$ by the finiteness of
$\hat
{v}(k)$. On the other hand, the third remaining term $ |\widehat{\Gamma}_{t}(k)\mathit{I}(t,(1-\epsilon)\widehat{\Gamma}_{t}(k)) |\mathbf
{1}_{\{
y_{1}<\widehat{\Gamma}_{t}(k)< {y_{2}}/{(1-\epsilon_{0})}\}}$ is
dominated by $k\widehat{\Gamma}_{t}(k)\mathbf{1}_{\{y_{1}<\widehat{\Gamma}_{t}(k)< {y_{2}}/{(1-\epsilon_{0})}\}}$ for a constant $k>0$, and it
is obviously integrable as well.

We can let $\epsilon\rightarrow0$ and apply dominated convergence
theorem and Fatou's lemma to obtain the stated inequality.

To show that the optimal solution $\widehat{\Gamma}_{t}(k)>0$ for all
$t\in
[0,T]$, we can choose an element $\Gamma\in\widetilde{\mathfrak{Y}}(k)$
and $\Gamma_t>0$ for all $t\in[0,T]$. Inequality (\ref{wu}) can be
rewritten as
\begin{eqnarray*}
\label{wuhan} 0&\geq&\mathbb{E} \biggl[\int_{0}^{T}
\bigl(\Gamma_{t}-\widehat{\Gamma}_{t}(k) \bigr)\mathit{I}
\bigl(t,\Gamma_{t}^{\epsilon}\bigr)\mathbf{1}_{\{\widehat{\Gamma}_{t}>0\}
}\,dt
\biggr]
\\
&&{} +\mathbb{E} \biggl[\int_{0}^{T} \bigl(
\Gamma_{t}-\widehat{\Gamma}_{t}(k) \bigr)\mathit{I}\bigl(t,
\Gamma_{t}^{\epsilon}\bigr)\mathbf{1}_{\{\widehat{\Gamma}_{t}=0\}}\,dt \biggr].
\end{eqnarray*}
Now suppose $\overline{\mathbb{P}}\{\widehat{\Gamma}_{t}(k)=0\}>0$. By the
uniform integrability of $ \{ ((\Gamma_{t}-\widehat{\Gamma}_{t}(k))\* \mathit{I}(t,\Gamma^{\epsilon}_{t}) )^{-}, \epsilon\in
(0,1) \}$, the second term of (\ref{wuhan}) goes to $\infty$ as
$\epsilon$ converges to $0$ since $\mathit{I}(t,0)=\infty$, and
$\Gamma
_{t}>0$ for all $t\in[0,T]$. Then we obtain the contradiction and hence
the conclusion holds.
\end{pf}

%
%le5.14 #&#
\begin{lemma}\label{lemma14}
Under the assumptions of Theorem \ref{main-2}, the auxiliary dual
value function $\hat{v}(k)$ is continuously differentiable on
$(0,\infty
)$, and
\[
-k\hat{v}'(k)=\mathbb{E} \biggl[\int_{0}^{T}
\widehat{\Gamma}_{t}(k)\mathit {I}\bigl(t,\widehat{\Gamma}_{t}(k)
\bigr)\,dt \biggr].
\]
\end{lemma}

\begin{pf}
In order to show $\hat{v}(k)$ is continuously differentiable, by the
convex property, it is enough to justify that its derivative exists on
$(0,\infty)$. Now fix $k>0$, and define the function
\[
h(s)\triangleq\mathbb{E} \biggl[\int_{0}^{T}
\mathit{V}\biggl(t,\frac
{s}{k}\widehat{\Gamma}_{t}(k)\biggr)\,dt
\biggr].
\]

This function is convex, and by optimality of $\widehat{\Gamma}(k)$ of
problem (\ref{hat-v}), we have $h(s)\geq\hat{v}(s)$ for all $s>0$
and $h(k)=\hat{v}(k)$. Again, convexity implies that
\[
\Delta^{-}h(k)\leq\Delta^{-}\hat{v}(k)\leq
\Delta^{+}\hat {v}(k)\leq \Delta^{+}h(k),
\]
where $\Delta^{+}$ and $\Delta^{-}$ denote right and left derivatives,
respectively. Now
\begin{eqnarray*}
\Delta^{+}h(k)&= & \lim_{\epsilon\rightarrow0}\frac{h(k+\epsilon
)-h(k)}{\epsilon}
\\
&= & \lim_{\epsilon\rightarrow0}\frac{1}{\epsilon}\mathbb{E} \biggl[\int
_{0}^{T} \biggl(\mathit{V}\biggl(t,
\frac{k+\epsilon}{k}\widehat{\Gamma}_{t}(k)\biggr)-\mathit{V}\bigl(t,\widehat{\Gamma}_{t}(k)\bigr) \biggr)\,dt \biggr]
\\
&\leq& \mathop{{\lim\inf}}_{\epsilon\rightarrow0} \biggl(-\frac
{1}{k\epsilon} \biggr)
\mathbb{E} \biggl[\int_{0}^{T}\epsilon\widehat{\Gamma}_{t}(k)\mathit{I}\biggl(t,\frac{k+\epsilon}{k}\widehat{\Gamma}_{t}(k)\biggr)\,dt \biggr]
\\
&= & -\frac{1}{k}\mathbb{E} \biggl[\int_{0}^{T}
\widehat{\Gamma}_{t}(k)\mathit {I}\bigl(t,\widehat{\Gamma}_{t}(k)
\bigr)\,dt \biggr]
\end{eqnarray*}
by monotone convergence theorem. Similarly, we get
\[
\Delta^{-}h(k)\geq\mathop{{\lim\sup}}_{\epsilon\rightarrow
0}\mathbb {E}
\biggl[-\int_{0}^{T}\widehat{\Gamma}_{t}(k)
\mathit{I}\biggl(t,\frac
{k-\epsilon
}{k}\widehat{\Gamma}_{t}(k)\biggr)\,dt
\biggr].
\]

We can follow the same reasoning as in Lemma \ref{lemma13} to show
the family $\{(\widehat{\Gamma}_{t}(k)\mathit{I}(t,\frac{k-\epsilon
}{k}\widehat{\Gamma}_{t}(k))),\epsilon\in(0,1)\}$ is uniformly integrable.
Dominated convergence theorem and Fatou's lemma deduce
\[
\Delta^{-}h(k)\geq-\frac{1}{k}\mathbb{E} \biggl[\int
_{0}^{T}\widehat{\Gamma}_{t}(k)\mathit{I}
\bigl(t,\widehat{\Gamma}_{t}(k)\bigr)\,dt \biggr],
\]
which completes the proof.
\end{pf}

%
%le5.15 #&#
\begin{lemma}\label{newaddlemma}
The auxiliary dual value function $\hat{v}(\cdot)$ has the asymptotic property
%
%e5.32 #&#
\begin{equation}
-\hat{v}'(0)=\infty,\qquad -\hat{v}'(\infty)=0.
\end{equation}
\end{lemma}

\begin{pf}
We first show $-\hat{v}'(0)=\infty$, and to this end, we claim that
%
%e5.33 #&#
\begin{equation}
\label{gai1} \hat{v}(0+)\geq\int_{0}^{T}
\mathit{V}(t,0+)\,dt.
\end{equation}
First, for any $k>0$, it follows by definition that
\begin{eqnarray*}
\hat{v}(k) &=& \mathbb{E} \biggl[\int_{0}^{T}
\mathit{V}\bigl(t,\widehat{\Gamma}_{t}(k)\bigr)\,dt \biggr]
\\
&=& \mathbb{E}
\biggl[\int_{0}^{T}\mathit{V}^{+}\bigl(t,
\widehat{\Gamma}_{t}(k)\bigr)\,dt \biggr]-\mathbb{E} \biggl[\int
_{0}^{T}\mathit {V}^{-}\bigl(t,\widehat{\Gamma}_{t}(k)\bigr)\,dt \biggr].
\end{eqnarray*}
Recall that $\widetilde{\mathfrak{Y}}(k)=k\widetilde{\mathfrak{Y}}(1)$,
and thus $\widehat{\Gamma}_t(k)=k\widehat{\Gamma}_t(1)$. Now by Fatou's lemma,
first, we have
%
%e5.34 #&#
\begin{equation}
\label{assis22} \lim_{k\rightarrow0}\mathbb{E} \biggl[\int
_{0}^{T}\mathit {V}^{+}\bigl(t,\widehat{\Gamma}_{t}(k)\bigr)\,dt \biggr]\geq\mathbb{E} \biggl[\int
_{0}^{T}\mathit {V}^{+}(t,0+)\,dt \biggr].
\end{equation}
On the other hand, similar to the proof of Lemma \ref{lemma9}, we
can show that
\[
\mathbb{E} \biggl[\int_{0}^{T}
\mathit{V}^{-}\bigl(t,\widehat{\Gamma}_{t}(1)\bigr)\,dt \biggr]<
\infty,
\]
and therefore, by the monotonicity of function $\mathit{V}^{-}(t,\cdot
)$ and the dominated convergence theorem, it follows that
\[
\lim_{k\rightarrow0}\mathbb{E} \biggl[\int_{0}^{T}
\mathit {V}^{-}\bigl(t,\widehat{\Gamma}_{t}(k)\bigr)\,dt \biggr]=
\mathbb{E} \biggl[\int_{0}^{T}\mathit
{V}^{-}(t,0+)\,dt \biggr],
\]
which together with (\ref{assis22}) imply that (\ref{gai1}) holds true.

Therefore, if $ \int_{0}^{T}\mathit{V}(t,0+)\,dt =\infty$, we have
$\hat
{v}(0+)=\infty$, and by convexity, it follows that $\hat
{v}'(0+)=-\infty$.

In the case $ \int_{0}^{T}\mathit{V}(t,0+)\,dt <\infty$, it is easy to
see that
\[
-\hat{v}(0+)\geq\lim_{k\rightarrow0}\frac{\hat{v}(0)-\hat
{v}(k)}{k}\geq\lim
_{k\rightarrow0}\frac{\int_{0}^{T}\mathit
{V}(t,0+)\,dt-\mathbb{E} [\int_{0}^{T}\mathit{V}(t,\widehat{\Gamma}_{t}(k))\,dt ]}{k},
\]
and hence we have
\begin{eqnarray*}
-\hat{v}(0+)&\geq& \lim_{k\rightarrow0}\frac{\mathbb{E} [\int_{0}^{T}\mathit{V}(t,0+)\,dt ]-\mathbb{E} [\int_{0}^{T}\mathit
{V}(t,\widehat{\Gamma}_{t}(k))\,dt ]}{k}
\\
&\geq& \lim_{k\rightarrow0}\mathbb{E} \biggl[\int_{0}^{T}
\widehat{\Gamma}_{t}(1)\mathit{I}\bigl(t,k\widehat{\Gamma}_{t}(1)
\bigr)\,dt \biggr]=\infty
\end{eqnarray*}
by the monotone convergence theorem.

As for $-\hat{v}'(\infty)=0$, since the function $-\hat{v}$ is concave
and increasing, there is a finite positive limit
\[
-\hat{v}'(\infty)\triangleq\lim_{k\rightarrow\infty}-\hat
{v}'(y).
\]
By the definition of Legendre--Fenchel transform, for any $y>0$,
\[
-\mathit{V}(t,y)\leq-\mathit{U}(t,x)+xy\qquad\mbox{for all } x>0,
\]
and then for any $\epsilon>0$, we always have
\begin{eqnarray*}
0&\leq&-\hat{v}'(\infty)
\\
&=&  \lim_{k\rightarrow\infty}
\frac{-\hat
{v}(k)}{k}=\lim_{k\rightarrow\infty}\frac{\mathbb{E} [\int_{0}^{T}-\mathit{V}(t,\widehat{\Gamma}_{t}(k))\,dt ]}{k}
\\
&\leq& \lim_{k\rightarrow\infty}\frac{\mathbb{E} [\int_{0}^{T}-\mathit{U}(t,\epsilon\tilde{w}_{t})\,dt ]}{k}+\lim_{k\rightarrow\infty}
\frac{ {\langle}\epsilon\tilde{w},\widehat{\Gamma}(k) {\rangle}}{k}.
\end{eqnarray*}
Now, recall that for each fixed $(x,z)\in\mathcal{H}$, there exists a
constant $\lambda(x,z)>0$ such that $\tilde{w}_{t}\in\widetilde
{\mathcal
{A}}(\frac{\bar{p}}{\lambda}x,\frac{\bar{p}}{\lambda}z)$, and by the
definition of $\widetilde{\mathfrak{Y}}(k)$, we can see the second term
above satisfies
\[
\lim_{k\rightarrow\infty}\frac{ {\langle}\epsilon\tilde
{w},\widehat{\Gamma}(k) {\rangle}}{k}\leq\lim_{k\rightarrow\infty}
\frac
{\epsilon
({\bar{p}}/{\lambda})k}{k}=\epsilon\frac{\bar{p}}{\lambda
}.
\]

As for the first term, we claim that $\mathbb{E} [\int_{0}^{T}-\mathit{U}(t,\epsilon\tilde{w}_{t})\,dt ]<\infty$ for each
fixed $\epsilon$ small enough. Without loss of generality, it is enough
to consider that $\epsilon<\bar{x}$, and we can apply Lemma \ref{Tech} again. Since there exists a constant $x_{0}$ such that for all
$t\in[0,T]$,
\[
\mathit{U}\biggl(t,\frac{\epsilon}{\bar{x}} x\biggr)>\biggl(\frac{\epsilon}{\bar
{x}}
\biggr)^{-{\gamma}/(1-\gamma)}\mathit{U}(t,x)\qquad\mbox{for } 0<x\leq x_{0},
\]
we will have
\begin{eqnarray*}
&&\mathbb{E} \biggl[\int_{0}^{T}-\mathit{U}(t,
\epsilon\tilde {w}_{t})\,dt \biggr]
\\
&&\qquad =\mathbb{E} \biggl[\int_{0}^{T}-\mathit{U}(t,
\epsilon\tilde {w}_{t})\mathbf {1}_{\{\bar{x}\tilde{w}_{t}> x_{0}\}}\,dt \biggr]+\mathbb{E}
\biggl[\int_{0}^{T}-\mathit{U}(t,\epsilon
\tilde{w}_{t})\mathbf{1}_{\{\bar
{x}\tilde
{w}_{t}\leq x_{0}\}}\,dt \biggr]
\\
&&\qquad \leq\mathbb{E} \biggl[\int_{0}^{T}-\mathit{U}(t,
\epsilon\tilde {w}_{t})\mathbf{1}_{\{\bar{x}\tilde{w}_{t}> x_{0}\}}\,dt \biggr]+\biggl(
\frac
{\epsilon}{\bar{x}}\biggr)^{-{\gamma}/(1-\gamma)}\mathbb{E} \biggl[\int
_{0}^{T}-\mathit{U}(t,\bar{x}\tilde{w}_{t})\,dt
\biggr]
\\
&&\qquad <\infty
\end{eqnarray*}
by the fact that $\tilde{w}_t\leq1$ for $t\in[0,T]$ and condition
(\ref{ass41}).

Hence, in conclusion,
\[
0\leq-\hat{v}'(\infty)= \lim_{k\rightarrow\infty}
\frac{-\hat
{v}(k)}{k}\leq\epsilon\frac{\bar{p}}{\lambda},
\]
and consequently, we have $-\hat{v}'(\infty)=0$ by letting $\epsilon$
go to $0$.
\end{pf}

%
%le5.16 #&#
\begin{lemma}\label{lemma15}
Under assumptions of Theorem \ref{main-2}, for any $(x,z)\in
\mathcal
{H}$, suppose $k$ satisfies $1=-\hat{v}'(k)$ where $\hat{v}(k)$ is the
value\vspace*{1pt} function of the auxiliary dual optimization\vspace*{2pt} problem (\ref{hat-v}). Then $\tilde{c}_{t}^{\ast}(x,z)\triangleq\mathit
{I}(t,\widehat{\Gamma}_{t}(k))$ is the unique (in the sense of $=$ under $\overline
{\mathbb
{P}}$ in $\mathbb{L}_{+}^{0}$) optimal\vspace*{1pt} solution to problem (\ref{tilde-u}). Moreover we have $\tilde{c}_{t}^{\ast}(x,z)>0, \mathbb
{P}$-a.s. for all $t\in[0,T]$.
\end{lemma}

\begin{pf}
Lemma \ref{lemma14} asserts that
\[
{\bigl\langle} \tilde{c}^{\ast}(x,z),\widehat{\Gamma}(k) {\bigr\rangle }=-k
\hat {v}'(k)=k.
\]

And for any $\Gamma\in\widetilde{\mathfrak{Y}}(k)$, by Lemma \ref{lemma13}, we have
\[
{\bigl\langle} \tilde{c}^{\ast}(x,z), \Gamma(k) {\bigr\rangle}\leq {
\bigl\langle} \tilde{c}^{\ast}(x,z), \widehat{\Gamma}(k) {\bigr\rangle }=k.
\]

Hence, we first get $\tilde{c}_{t}^{\ast}(x,z)\in\widetilde
{\mathcal
{A}}(x,z)$ by the Bipolar relationship (\ref{eqnbipolar2}).

Now, for any $\tilde{c}\in\widetilde{\mathcal{A}}(x,z)$, we have
\begin{eqnarray*}
{\bigl\langle} \tilde{c}, \widehat{\Gamma}(k) {\bigr\rangle}&\leq& k,
\\
\mathit{U}(t,\tilde{c}_{t})&\leq& \mathit{V}\bigl(t,\widehat{\Gamma}_{t}(k)\bigr)+\tilde{c}_{t}\widehat{\Gamma}_{t}(k)\qquad
\forall t\in[0,T].
\end{eqnarray*}

It follows that
%
%e5.35 #&#
\begin{eqnarray}
&& \mathbb{E} \biggl[\int_{0}^{T}\mathit{U}(t,
\tilde{c}_{t})\,dt \biggr]\nonumber
\\
&&\qquad \leq  \hat {v}(k)+k
= \mathbb{E} \biggl[\int
_{0}^{T} \bigl(\mathit{V}\bigl(t,\widehat{\Gamma}_{t}(k)\bigr)+\widehat{\Gamma}_{t}(k)\mathit{I}\bigl(t,\widehat{\Gamma}_{t}(k)\bigr) \bigr)\,dt \biggr]
\\
&&\qquad =  \mathbb{E} \biggl[\int_{0}^{T}\mathit{U}
\bigl(t,\mathit{I}\bigl(\widehat{\Gamma}_{t}(k)\bigr)\bigr)\,dt \biggr]=
\mathbb{E} \biggl[\int_{0}^{T}\mathit{U}\bigl(t,
\tilde {c}_{t}^{\ast}\bigr)\,dt \biggr],
\nonumber
\end{eqnarray}
which infers the optimality of $\tilde{c}^{\ast}$. The uniqueness of
the optimal solution follows from the strict concavity of the function
$\mathit{U}$.

Under assumptions of Theorem \ref{main-2}, for any pair $(x,z)\in
\mathcal{H}$, because $\widetilde{\mathfrak{Y}}(k)$ is convexly compact
and $\widehat{\Gamma}_{t}(k)$ is bounded in probability, we have the
optimal solution $\tilde{c}_{t}^{\ast}(x,z)>0, \mathbb{P}$-a.s. for
all\vspace*{1pt} $t\in[0,T]$ since $\widehat{\Gamma}_{t}(k)$ is bounded in probability
if and only if $\widehat{\Gamma}_{t}(k)$ is finite $\overline{\mathbb{P}}$-a.s.
and by definition, we know $\mathit{I}(t,x)>0$ for $x<\infty$.
\end{pf}

Let $(x,z)\in\operatorname{cl}\mathcal{H}$, the proof of Lemma \ref{lemma3} shows
that there exists $\tilde{c}\in\widetilde{\mathcal{A}}(x,z)$ such that
$\overline{\mathbb{P}}[\tilde{c}>0]>0$. Similar to the proof of Lemma $12$
of Hugonnier and Kramkov \cite{kram04}, we will have:

%le5.17 #&#
\begin{lemma}\label{lemma16}
Assume that conditions of Proposition \ref{Prop1} hold, and let
$(y^{n},r^{n})\in\mathcal{R}$ and $\Gamma^{n}\in\widetilde
{\mathcal
{Y}}(y^{n},r^{n})$, $n\geq1$, converge to $(y,r)\in\mathbb{R}^{2}$ and
$\Gamma\in\mathbb{L}_{+}^{0}$, respectively. If $\Gamma$ is a strictly
positive random variable, we have $(y,r)\in\mathcal{R}$ and $\Gamma
\in
\widetilde{\mathcal{Y}}(y,r)$.
\end{lemma}

The next lemma is the last result we need to prepare to proceed to the
proof of Theorem \ref{main-2}:

%le5.18 #&#
\begin{lemma}\label{lemequc}
Under Assumption \ref{ass5}, we have
\[
\overline{\mathbb{P}}\bigl[\tilde{c}^{\ast}(x_1,z_1)
\neq\tilde{c}^{\ast
}(x_2,z_2)\bigr]>0,
\]
for two different points $(x_i,z_i)\in\mathcal{H}$, $i=1,2$.
\end{lemma}

\begin{pf}
Assume that there exist two distinct pairs $(x_1,z_1)$ and $(x_2,z_2)$
in $\mathcal{H}$ and
\[
\overline{\mathbb{P}}\bigl[\tilde{c}^{\ast}(x_1,z_1)
\neq\tilde{c}^{\ast
}(x_2,z_2)\bigr]=0.
\]
The definition of set $\widetilde{A}(x_1,z_1)$ implies that
\[
\bigl\langle\tilde{c}^{\ast}(x_2,z_2), \Gamma
\bigr\rangle\leq x_1-z_1 \langle\tilde{w}, \Gamma
\rangle\qquad \forall\Gamma \in \widetilde{\mathcal{M}}.
\]
However, we also know that $\tilde{c}^{\ast}(x_2, z_2)\in\widetilde
{A}(x_2, z_2)$, which deduces that
\[
x_2-z_2 \langle\tilde{w}, \Gamma \rangle\leq
x_1-z_1 \langle\tilde{w}, \Gamma \rangle\qquad \forall\Gamma
\in\widetilde {\mathcal{M}}.
\]
On the other hand, by symmetry and replacing $\tilde{c}^{\ast
}(x_2,z_2)$ by $\tilde{c}^{\ast}(x_1,z_1)$, we can conclude that
\[
x_1-z_1 \langle\tilde{w}, \Gamma \rangle\leq
x_2-z_2 \langle \tilde{w}, \Gamma \rangle\qquad \forall
\Gamma\in\widetilde {\mathcal {M}}.
\]
Therefore we must have
\[
x_1-z_1 \langle\tilde{w}, \Gamma \rangle=
x_2-z_2 \langle \tilde{w}, \Gamma \rangle\qquad \forall
\Gamma\in\widetilde {\mathcal {M}},
\]
which is a contradiction to Assumption \ref{ass5} since we can obtain
a constant $K=\frac{x_1-x_2}{z_1-z_2}$ and
\[
\mathbb{E}^{\mathbb{Q}}[\mathcal{E}]= \langle w, Y \rangle = \langle\tilde{w},
\Gamma \rangle=\frac{x_1-x_2}{z_1-z_2}\qquad \forall\mathbb{Q}\in\mathcal{M}.
\]\upqed
\end{pf}

\begin{pf*}{Proof of Theorem \ref{main-2}}
We first show that the dual value function $\tilde{v}(y,z)$ is
continuously differentiable on $\mathcal{R}$. Theorems $4.1.1$ and
$4.1.2$ in Hiriart-Urruty and Lemar\'{e}chal \cite{MR1865628} give the
equivalence between the above statement and the fact that the value
function $\tilde{u}(x,z)$ is strictly concave on $\mathcal{H}$, since
$\mathit{U}$ is a strictly concave function. Showing that the value
function is strictly concave is equivalent to showing that for any two
distinct points $(x_{i},z_{i})\in\mathcal{H}$, $i=1,2$, the optimal
consumption policies are different, that is,
\[
\overline{\mathbb{P}}\bigl[\tilde{c}^{\ast}(x_{1},z_{1})
\neq\tilde{c}^{\ast
}(x_{2},z_{2})\bigr]>0,
\]
which is the consequence of Lemma \ref{lemequc}.

To continue the remaining part, it amounts to show that assertion
(ii) holds. Recall that $\widehat{\Gamma}(k)$ is the optimal solution of
the auxiliary dual problem (\ref{hat-v}) such that
\[
\widehat{\Gamma}_{t}(k)=\mathit{U}'\bigl(t,
\tilde{c}_{t}^{\ast}(x,z)\bigr)\qquad \forall t\in[0,T], k= {\bigl
\langle} \tilde{c}^{\ast}(x,z),\widehat{\Gamma}(k) {\bigr\rangle}.
\]
As\vspace*{2pt} $\widetilde{\mathfrak{Y}}(k)$ is closed under convergence in measure
$\overline{\mathbb{P}}$, there exists a sequence $(y^{n},r^{n})\in
k\mathfrak
{R}(x,z)$ such that $\Gamma^{n}\in\widetilde{\mathcal
{Y}}(y^{n},r^{n})$, and $\Gamma^{n}$ converges to $\widehat{\Gamma}(k)$
$\overline{\mathcal{P}}$-a.s. by passing to a subsequence if necessary.
Since the set $k\mathfrak{R}(x,z)$ is bounded, there exists a further
subsequence $(y^{n},r^{n})$ converging\vspace*{2pt} to $(y,r)\in\mathbb{R}^{2}$. By
passing to this further subsequence, as we have shown $\overline{\mathbb
{P}}[\widehat{\Gamma}(k) >0]=1$, it follows that $(y,r)\in k\mathfrak
{R}(x,z)$ such that $\widehat{\Gamma}(k)\in\widetilde{\mathcal{Y}}(y,r)$
due to Lemma \ref{lemma16}. Moreover, for this pair $(y,r)\in
\mathcal
{R}$, by Fatou's lemma and Proposition \ref{Prop1}, we have the
equality that
%
%e5.36 #&#
\begin{equation}
\label{Th2-1} xy-zr=k= {\bigl\langle} \tilde{c}^{\ast}(x,z),\widehat{\Gamma}(k) {\bigr\rangle}.
\end{equation}
The corresponding optimizer $\Gamma_{t}^{\ast}(y,r)$ of (\ref{tilde-v}) then verifies
%
%e5.37 #&#
\begin{equation}
\label{Th2-2} \Gamma_{t}^{\ast}(y,r)=\widehat{\Gamma}_{t}(k)=\mathit{U}'\bigl(t,\tilde
{c}^{\ast
}(x,z)\bigr),\qquad t\in[0,T].
\end{equation}
To see this, on one hand, we have $\widehat{\Gamma}(k)\in\widetilde
{\mathcal
{Y}}(y,r)$, hence
\begin{eqnarray*}
\mathbb{E} \biggl[\int_{0}^{T}\mathit{V}\bigl(t,
\Gamma_{t}^{\ast
}(y,r)\bigr)\,dt \biggr]&=&\inf
_{\Gamma\in\widetilde{\mathcal{Y}}(y,r)}\mathbb{E} \biggl[\int_{0}^{T}
\mathit{V}\bigl(t,\Gamma_{t}(y,r)\bigr)\,dt \biggr]
\\
&\leq&\mathbb{E} \biggl[
\int_{0}^{T}\mathit{V}\bigl(t,\widehat{\Gamma}_{t}(k)\bigr)\,dt \biggr].
\end{eqnarray*}
On the other hand, we have
\begin{eqnarray*}
\mathbb{E} \biggl[\int_{0}^{T}\mathit{V}\bigl(t,
\widehat{\Gamma}_{t}(y,r)\bigr)\,dt \biggr]&= & \inf_{\Gamma\in\widetilde{\mathfrak{Y}}(k)}
\mathbb{E} \biggl[\int_{0}^{T}\mathit{V}\bigl(t,
\Gamma_{t}(y,r)\bigr)\,dt \biggr]
\\
&\leq& \inf_{\Gamma\in\widetilde{\mathcal{Y}}(y,r)}\mathbb{E} \biggl[\int_{0}^{T}
\mathit{V}\bigl(t,\Gamma_{t}(y,r)\bigr)\,dt \biggr]
\\
&=& \mathbb{E} \biggl[
\int_{0}^{T}\mathit{V}\bigl(t,\Gamma_{t}^{\ast}(y,r)
\bigr)\,dt \biggr].
\end{eqnarray*}

By the equality
\[
\mathit{U}\bigl(t,\tilde{c}_{t}^{\ast}(x,z)\bigr)=\mathit{V}
\bigl(t,\widehat{\Gamma}_{t}(k)\bigr)+\tilde{c}_{t}^{\ast}(x,z)
\widehat{\Gamma}_{t}(k),\qquad t\in [0,T],
\]
we can conclude $(y,r)\in\partial\tilde{u}(x,z)$ by Theorem $23.5$ of
Rockafellar \cite{Rock}, since
%
%e5.38 #&#
\begin{equation}
\label{Th2-3} \tilde{u}(x,z)=\tilde{v}(y,z)+xy-zr.
\end{equation}
In particular, it implies that
%
%e5.39 #&#
\begin{equation}
\label{Th2-4} \partial\tilde{u}(x,z)\cap\mathcal{R}\neq\varnothing.
\end{equation}

Similar to the proof of Theorem $2$ in Hugonnier and Kramkov \cite
{kram04}, it is easy to show that
\[
\partial\tilde{u}(x,z)\subset\mathcal{R}.
\]
For any $(y,r)\in\partial\tilde{u}(x,z)$, there exists a sequence
$(y^{n},r^{n})\in\partial\tilde{u}(x,z)\cap\mathcal{R}$ converging to
$(y,r)$ by (\ref{Th2-4}) and the fact that $\partial\tilde{u}(x,z)$
is closed and convex. Since $\mathit{U}'(\cdot,\tilde{c}_{\cdot
}^{\ast
}(x,z))$ is strictly positive and $\mathit{U}'(\cdot,\tilde
{c}_{\cdot
}^{\ast}(x,z))\in\widetilde{\mathcal{Y}}(y,r)$, Lemma \ref{lemma16}
infers that $(y,r)\in\mathcal{R}$.

Conversely, for any $(y,r)\in\partial\tilde{u}(x,z)$, we have
\begin{eqnarray*}
&& \mathbb{E} \biggl[\int_{0}^{T} \bigl|\mathit{V}
\bigl(t,\Gamma_{t}^{\ast
}(y,r)\bigr)+\tilde{c}_{t}^{\ast}(x,z)
\Gamma_{t}^{\ast}(y,r)-\mathit {U}\bigl(t,\tilde{c}_{t}^{\ast}(x,z)
\bigr) \bigr|\,dt \biggr]
\\
&&\qquad =  \mathbb{E} \biggl[ \biggl(\int_{0}^{T}
\mathit{V}\bigl(t,\Gamma_{t}^{\ast
}(y,r)\bigr)+
\tilde{c}_{t}^{\ast}(x,z)\Gamma_{t}^{\ast}(y,r)-
\mathit {U}\bigl(t,\tilde{c}_{t}^{\ast}(x,z)\bigr)\,dt \biggr)
\biggr]
\\
&&\qquad \leq \tilde{v}(y,r)+xy-zr-\tilde{u}(x,z)=0,
\end{eqnarray*}
which infers (\ref{Th2-1}) and (\ref{Th2-2}).
\end{pf*}

\section*{Acknowledgement}
I sincerely thank my advisor Mihai S\^{\i
}rbu for numerous helpful discussions on the topic of this paper as
well as his generous support and kind encouragement during my research.


%suskaldyti doi

% imsref loaded by linak, 2014-05-26 09:09:55
% imsref loaded by linak, 2014-05-27 11:09:46

\printaddresses

\begin{thebibliography}{31}

%b1 #&#
\bibitem{pham}
\begin{barticle}[mr]
\bauthor{\bsnm{Bouchard},~\bfnm{Bruno}\binits{B.}} \AND
\bauthor{\bsnm{Pham},~\bfnm{Huy{\^e}n}\binits{H.}}
(\byear{2004}).
\btitle{Wealth-path dependent utility maximization in incomplete markets}.
\bjournal{Finance Stoch.}
\bvolume{8}
\bpages{579--603}.
\bid{doi={10.1007/s00780-004-0125-8}, issn={0949-2984}, mr={2212119}}
\end{barticle}
\bptok{imsref}%
% NOT OUTPUTED:
% issn = 0949-2984
% url = http://dx.doi.org/10.1007/s00780-004-0125-8
% number = 4
% fjournal = Finance and Stochastics
\endbibitem

%b2 #&#
\bibitem{Brannath}
\begin{bincollection}[mr]
\bauthor{\bsnm{Brannath},~\bfnm{W.}\binits{W.}} \AND
\bauthor{\bsnm{Schachermayer},~\bfnm{W.}\binits{W.}}
(\byear{1999}).
\btitle{A bipolar theorem for {$L^ 0_ +(\Omega,\mathcal F,\mathbf P)$}}.
In \bbooktitle{S\'eminaire de {P}robabilit\'es, {XXXIII}}.
\bseries{Lecture Notes in Math.}
\bvolume{1709}
\bpages{349--354}.
\bpublisher{Springer},
\blocation{Berlin}.
\bid{doi={10.1007/BFb0096525}, mr={1768009}}
\end{bincollection}
\bptok{imsref}%
% NOT OUTPUTED:
% url = http://dx.doi.org/10.1007/BFb0096525
\endbibitem

%b3 #&#
\bibitem{Campbell}
\begin{barticle}[auto:STB|2014/02/12|14:17:21]
\bauthor{\bsnm{Campbell},~\bfnm{J.~Y.}\binits{J.~Y.}} \AND
\bauthor{\bsnm{Cochrane},~\bfnm{J.~H.}\binits{J.~H.}}
(\byear{1999}).
\btitle{By force of habit: A consumption-based explanation of aggregate stock market behavior}.
\bjournal{J. Polit. Econ.}
\bvolume{107}
\bpages{205--251}.
\end{barticle}
\bptok{imsref}%
\endbibitem

%b4 #&#
\bibitem{constantinides1988habit}
\begin{bmisc}[auto:STB|2014/02/12|14:17:21]
\bauthor{\bsnm{Constantinides},~\bfnm{G.~M.}\binits{G.~M.}}
(\byear{1988}).
\bhowpublished{Habit formation: A resolution of the equity premium puzzle.
Working paper series. Center for Research in Security Prices, Graduate School of Business,
Univ. Chicago, Chicago, IL.}
\end{bmisc}
\bptok{imsref}%
\endbibitem

%b5 #&#
\bibitem{Cvitanic}
\begin{barticle}[mr]
\bauthor{\bsnm{Cvitani{\'c}},~\bfnm{Jak{\v{s}}a}\binits{J.}},
\bauthor{\bsnm{Schachermayer},~\bfnm{Walter}\binits{W.}} \AND
\bauthor{\bsnm{Wang},~\bfnm{Hui}\binits{H.}}
(\byear{2001}).
\btitle{Utility maximization in incomplete markets with random endowment}.
\bjournal{Finance Stoch.}
\bvolume{5}
\bpages{259--272}.
\bid{doi={10.1007/PL00013534}, issn={0949-2984}, mr={1841719}}
\end{barticle}
\bptok{imsref}%
% NOT OUTPUTED:
% issn = 0949-2984
% url = http://dx.doi.org/10.1007/PL00013534
% number = 2
% fjournal = Finance and Stochastics
\endbibitem

%b6 #&#
\bibitem{Schachermayer94}
\begin{barticle}[mr]
\bauthor{\bsnm{Delbaen},~\bfnm{Freddy}\binits{F.}} \AND
\bauthor{\bsnm{Schachermayer},~\bfnm{Walter}\binits{W.}}
(\byear{1994}).
\btitle{A general version of the fundamental theorem of asset pricing}.
\bjournal{Math. Ann.}
\bvolume{300}
\bpages{463--520}.
\bid{doi={10.1007/BF01450498}, issn={0025-5831}, mr={1304434}}
\end{barticle}
\bptok{imsref}%
% NOT OUTPUTED:
% issn = 0025-5831
% url = http://dx.doi.org/10.1007/BF01450498
% number = 3
% coden = MAANA
% fjournal = Mathematische Annalen
\endbibitem

%b7 #&#
\bibitem{schachermayer98}
\begin{barticle}[mr]
\bauthor{\bsnm{Delbaen},~\bfnm{F.}\binits{F.}} \AND
\bauthor{\bsnm{Schachermayer},~\bfnm{W.}\binits{W.}}
(\byear{1998}).
\btitle{The fundamental theorem of asset pricing for unbounded stochastic processes}.
\bjournal{Math. Ann.}
\bvolume{312}
\bpages{215--250}.
\bid{doi={10.1007/s002080050220}, issn={0025-5831}, mr={1671792}}
\end{barticle}
\bptok{imsref}%
% NOT OUTPUTED:
% issn = 0025-5831
% url = http://dx.doi.org/10.1007/s002080050220
% number = 2
% coden = MAANA
% fjournal = Mathematische Annalen
\endbibitem

%b8 #&#
\bibitem{Detemple2003265}
\begin{barticle}[mr]
\bauthor{\bsnm{Detemple},~\bfnm{J{\'e}r{\^o}me~B.}\binits{J.~B.}} \AND
\bauthor{\bsnm{Karatzas},~\bfnm{Ioannis}\binits{I.}}
(\byear{2003}).
\btitle{Non-addictive habits: Optimal consumption-portfolio policies}.
\bjournal{J. Econom. Theory}
\bvolume{113}
\bpages{265--285}.
\bid{doi={10.1016/S0022-0531(03)00099-1}, issn={0022-0531}, mr={2021304}}
\end{barticle}
\bptok{imsref}%
% NOT OUTPUTED:
% issn = 0022-0531
% url = http://dx.doi.org/10.1016/S0022-0531(03)00099-1
% number = 2
% coden = JECTAQ
% fjournal = Journal of Economic Theory
\endbibitem

%b9 #&#
\bibitem{detemple91}
\begin{barticle}[auto:STB|2014/02/12|14:17:21]
\bauthor{\bsnm{Detemple},~\bfnm{J.~B.}\binits{J.~B.}} \AND
\bauthor{\bsnm{Zapatero},~\bfnm{F.}\binits{F.}}
(\byear{1991}).
\btitle{Asset prices in an exchange economy with habit formation}.
\bjournal{Econometrica}
\bvolume{59}
\bpages{1633--1657}.
\end{barticle}
\bptok{imsref}%
% NOT OUTPUTED:
% number = 6
\endbibitem

%b10 #&#
\bibitem{detemple92}
\begin{barticle}[auto:STB|2014/02/12|14:17:21]
\bauthor{\bsnm{Detemple},~\bfnm{J.~B.}\binits{J.~B.}} \AND
\bauthor{\bsnm{Zapatero},~\bfnm{F.}\binits{F.}}
(\byear{1992}).
\btitle{Optimal consumption-portfolio policies with habit formation}.
\bjournal{Math. Finance}
\bvolume{2}
\bpages{251--274}.
\end{barticle}
\bptok{imsref}%
% NOT OUTPUTED:
% number = 4
\endbibitem

%b11 #&#
\bibitem{Eng09}
\begin{barticle}[mr]
\bauthor{\bsnm{Englezos},~\bfnm{Nikolaos}\binits{N.}} \AND
\bauthor{\bsnm{Karatzas},~\bfnm{Ioannis}\binits{I.}}
(\byear{2009}).
\btitle{Utility maximization with habit formation: Dynamic programming and stochastic {PDE}s}.
\bjournal{SIAM J. Control Optim.}
\bvolume{48}
\bpages{481--520}.
\bid{doi={10.1137/070686998}, issn={0363-0129}, mr={2486081}}
\end{barticle}
\bptok{imsref}%
% NOT OUTPUTED:
% issn = 0363-0129
% url = http://dx.doi.org/10.1137/070686998
% number = 2
% fjournal = SIAM Journal on Control and Optimization
\endbibitem

%b12 #&#
\bibitem{Hicks}
\begin{bbook}[auto:STB|2014/02/12|14:17:21]
\bauthor{\bsnm{Hicks},~\bfnm{J.}\binits{J.}}
(\byear{1965}).
\btitle{Capital and Growth}.
\bpublisher{Oxford Univ. Press},
\blocation{New York}.
\end{bbook}
\bptok{imsref}%
\endbibitem

%b13 #&#
\bibitem{MR1865628}
\begin{bbook}[mr]
\bauthor{\bsnm{Hiriart-Urruty},~\bfnm{Jean-Baptiste}\binits{J.-B.}} \AND
\bauthor{\bsnm{Lemar{\'e}chal},~\bfnm{Claude}\binits{C.}}
(\byear{2001}).
\btitle{Fundamentals of Convex Analysis}.
\bpublisher{Springer},
\blocation{Berlin}.
\bid{doi={10.1007/978-3-642-56468-0}, mr={1865628}}
\end{bbook}
\bptok{imsref}%
% NOT OUTPUTED:
% isbn = 3-540-42205-6
% url = http://dx.doi.org/10.1007/978-3-642-56468-0
% fpage = x+259
\endbibitem

%b14 #&#
\bibitem{kram04}
\begin{barticle}[mr]
\bauthor{\bsnm{Hugonnier},~\bfnm{Julien}\binits{J.}} \AND
\bauthor{\bsnm{Kramkov},~\bfnm{Dmitry}\binits{D.}}
(\byear{2004}).
\btitle{Optimal investment with random endowments in incomplete markets}.
\bjournal{Ann. Appl. Probab.}
\bvolume{14}
\bpages{845--864}.
\bid{doi={10.1214/105051604000000134}, issn={1050-5164}, mr={2052905}}
\end{barticle}
\bptok{imsref}%
% NOT OUTPUTED:
% issn = 1050-5164
% url = http://dx.doi.org/10.1214/105051604000000134
% number = 2
% fjournal = The Annals of Applied Probability
\endbibitem

%b15 #&#
\bibitem{KKK}
\begin{barticle}[mr]
\bauthor{\bsnm{Karatzas},~\bfnm{Ioannis}\binits{I.}},
\bauthor{\bsnm{Lehoczky},~\bfnm{John~P.}\binits{J.~P.}},
\bauthor{\bsnm{Shreve},~\bfnm{Steven~E.}\binits{S.~E.}} \AND
\bauthor{\bsnm{Xu},~\bfnm{Gan-Lin}\binits{G.-L.}}
(\byear{1991}).
\btitle{Martingale and duality methods for utility maximization in an incomplete market}.
\bjournal{SIAM J. Control Optim.}
\bvolume{29}
\bpages{702--730}.
\bid{doi={10.1137/0329039}, issn={0363-0129}, mr={1089152}}
\end{barticle}
\bptok{imsref}%
% NOT OUTPUTED:
% issn = 0363-0129
% url = http://dx.doi.org/10.1137/0329039
% number = 3
% coden = SJCODE
% fjournal = SIAM Journal on Control and Optimization
\endbibitem

%b16 #&#
\bibitem{KarZit03}
\begin{barticle}[mr]
\bauthor{\bsnm{Karatzas},~\bfnm{Ioannis}\binits{I.}} \AND
\bauthor{\bsnm{{\v{Z}}itkovi{\'c}},~\bfnm{Gordan}\binits{G.}}
(\byear{2003}).
\btitle{Optimal consumption from investment and random endowment in incomplete semimartingale markets}.
\bjournal{Ann. Probab.}
\bvolume{31}
\bpages{1821--1858}.
\bid{doi={10.1214/aop/1068646367}, issn={0091-1798}, mr={2016601}}
\end{barticle}
\bptok{imsref}%
% NOT OUTPUTED:
% issn = 0091-1798
% url = http://dx.doi.org/10.1214/aop/1068646367
% number = 4
% coden = APBYAE
% fjournal = The Annals of Probability
\endbibitem

%b17 #&#
\bibitem{Helena}
\begin{bbook}[mr]
\bauthor{\bsnm{Kauppila},~\bfnm{Helena}\binits{H.}}
(\byear{2010}).
\btitle{Convex Duality in Singular Control---{O}ptimal Consumption Choice with Intertemporal Substitution and Optimal Investment in Incomplete Markets}.
\bpublisher{ProQuest LLC},
\blocation{Ann Arbor, MI}.
\bnote{Ph.D. thesis---Columbia Univ.}
\bid{mr={2733498}}
\end{bbook}
\bptok{imsref}%
% NOT OUTPUTED:
% isbn = 978-1109-67376-0
% url = http://gateway.proquest.com/openurl?url_ver=Z39.88-2004&rft_val_fmt=info:ofi/fmt:kev:mtx:dissertation&res_dat=xri:pqdiss&rft_dat=xri:pqdiss:3400545
% fpage = 126
\endbibitem

%b18 #&#
\bibitem{kram99}
\begin{barticle}[mr]
\bauthor{\bsnm{Kramkov},~\bfnm{D.}\binits{D.}} \AND
\bauthor{\bsnm{Schachermayer},~\bfnm{W.}\binits{W.}}
(\byear{1999}).
\btitle{The asymptotic elasticity of utility functions and optimal investment in incomplete markets}.
\bjournal{Ann. Appl. Probab.}
\bvolume{9}
\bpages{904--950}.
\bid{doi={10.1214/aoap/1029962818}, issn={1050-5164}, mr={1722287}}
\end{barticle}
\bptok{imsref}%
% NOT OUTPUTED:
% issn = 1050-5164
% url = http://dx.doi.org/10.1214/aoap/1029962818
% number = 3
% fjournal = The Annals of Applied Probability
\endbibitem

%b19 #&#
\bibitem{kram03}
\begin{barticle}[mr]
\bauthor{\bsnm{Kramkov},~\bfnm{D.}\binits{D.}} \AND
\bauthor{\bsnm{Schachermayer},~\bfnm{W.}\binits{W.}}
(\byear{2003}).
\btitle{Necessary and sufficient conditions in the problem of optimal investment in incomplete markets}.
\bjournal{Ann. Appl. Probab.}
\bvolume{13}
\bpages{1504--1516}.
\bid{doi={10.1214/aoap/1069786508}, issn={1050-5164}, mr={2023886}}
\end{barticle}
\bptok{imsref}%
% NOT OUTPUTED:
% issn = 1050-5164
% url = http://dx.doi.org/10.1214/aoap/1069786508
% number = 4
% fjournal = The Annals of Applied Probability
\endbibitem

%b20 #&#
\bibitem{MR1402653}
\begin{barticle}[mr]
\bauthor{\bsnm{Kramkov},~\bfnm{D.~O.}\binits{D.~O.}}
(\byear{1996}).
\btitle{Optional decomposition of supermartingales and hedging contingent claims in incomplete security markets}.
\bjournal{Probab. Theory Related Fields}
\bvolume{105}
\bpages{459--479}.
\bid{doi={10.1007/BF01191909}, issn={0178-8051}, mr={1402653}}
\end{barticle}
\bptok{imsref}%
% NOT OUTPUTED:
% issn = 0178-8051
% url = http://dx.doi.org/10.1007/BF01191909
% number = 4
% coden = PTRFEU
% fjournal = Probability Theory and Related Fields
\endbibitem

%b21 #&#
\bibitem{Mehra}
\begin{barticle}[auto:STB|2014/02/12|14:17:21]
\bauthor{\bsnm{Mehra},~\bfnm{R.}\binits{R.}} \AND
\bauthor{\bsnm{Prescott},~\bfnm{E.~C.}\binits{E.~C.}}
(\byear{1985}).
\btitle{The equity premium: A puzzle}.
\bjournal{J. Monet. Econ.}
\bvolume{15}
\bpages{145--161}.
\end{barticle}
\bptok{imsref}%
% NOT OUTPUTED:
% number = 2
\endbibitem

%b22 #&#
\bibitem{Munk}
\begin{barticle}[mr]
\bauthor{\bsnm{Munk},~\bfnm{Claus}\binits{C.}}
(\byear{2008}).
\btitle{Portfolio and consumption choice with stochastic investment opportunities and habit formation in preferences}.
\bjournal{J. Econom. Dynam. Control}
\bvolume{32}
\bpages{3560--3589}.
\bid{doi={10.1016/j.jedc.2008.02.005}, issn={0165-1889}, mr={2464350}}
\end{barticle}
\bptok{imsref}%
% NOT OUTPUTED:
% issn = 0165-1889
% url = http://dx.doi.org/10.1016/j.jedc.2008.02.005
% number = 11
% coden = JEDCDH
% fjournal = Journal of Economic Dynamics \& Control
\endbibitem

%b23 #&#
\bibitem{Rock}
\begin{bbook}[mr]
\bauthor{\bsnm{Rockafellar},~\bfnm{R.~Tyrrell}\binits{R.~T.}}
(\byear{1970}).
\btitle{Convex Analysis}.
\bpublisher{Princeton Univ. Press},
\blocation{Princeton, NJ}.
\bid{mr={0274683}}
\end{bbook}
\bptok{imsref}%
% NOT OUTPUTED:
% fpage = xviii+451
\endbibitem

%b24 #&#
\bibitem{Heal}
\begin{barticle}[auto:STB|2014/02/12|14:17:21]
\bauthor{\bsnm{Ryder},~\bfnm{H.~E.}\binits{H.~E.}} \AND
\bauthor{\bsnm{Heal},~\bfnm{G.~M.}\binits{G.~M.}}
(\byear{1973}).
\btitle{Optimal growth with intertemporally dependent preferences}.
\bjournal{Rev. Econom. Stud.}
\bvolume{40}
\bpages{1--33}.
\end{barticle}
\bptok{imsref}%
\endbibitem

%b25 #&#
\bibitem{RePEccc}
\begin{barticle}[auto:STB|2014/02/12|14:17:21]
\bauthor{\bsnm{Samuelson},~\bfnm{P.~A.}\binits{P.~A.}}
(\byear{1969}).
\btitle{Lifetime portfolio selection by dynamic stochastic programming}.
\bjournal{Rev. Econ. Stat.}
\bvolume{51}
\bpages{239--246}.
\end{barticle}
\bptok{imsref}%
% NOT OUTPUTED:
% number = 3
\endbibitem

%b26 #&#
\bibitem{MR2113724}
\begin{bbook}[mr]
\bauthor{\bsnm{Schachermayer},~\bfnm{Walter}\binits{W.}}
(\byear{2004}).
\btitle{Portfolio Optimization in Incomplete Financial Markets}.
\bpublisher{Scuola Normale Superiore},
\blocation{Classe di Scienze, Pisa}.
\bid{mr={2144570}}
\end{bbook}
\bptok{imsref}%
% NOT OUTPUTED:
% isbn = 88-7642-141-6
% fpage = vi+65
\endbibitem

%b27 #&#
\bibitem{Schroder01072002}
\begin{barticle}[auto:STB|2014/02/12|14:17:21]
\bauthor{\bsnm{Schroder},~\bfnm{M.}\binits{M.}} \AND
\bauthor{\bsnm{Skiadas},~\bfnm{C.}\binits{C.}}
(\byear{2002}).
\btitle{An isomorphism between asset pricing models with and without linear habit formation}.
\bjournal{Rev. Financ. Stud.}
\bvolume{15}
\bpages{1189--1221}.
\end{barticle}
\bptok{imsref}%
% NOT OUTPUTED:
% number = 4
\endbibitem

%b28 #&#
\bibitem{Yu}
\begin{bmisc}[auto]
\bauthor{\bsnm{Yu},~\bfnm{Xiang}\binits{X.}}
(\byear{2012}).
\bhowpublished{Utility {m}aximization with {c}onsumption {h}abit {f}ormation in {i}ncomplete {m}arkets.
Ph.D. thesis, Univ. Texas at Austin.}
\end{bmisc}
\bptok{imsref}%
% NOT OUTPUTED:
% isbn = 978-1303-41023-9
% url = http://gateway.proquest.com/openurl?url_ver=Z39.88-2004&rft_val_fmt=info:ofi/fmt:kev:mtx:dissertation&res_dat=xri:pqm&rft_dat=xri:pqdiss:3572868
% fpage = 182
\endbibitem

%b29 #&#
\bibitem{Zit02}
\begin{barticle}[mr]
\bauthor{\bsnm{{\v{Z}}itkovi{\'c}},~\bfnm{Gordan}\binits{G.}}
(\byear{2002}).
\btitle{A filtered version of the bipolar theorem of {B}rannath and {S}chachermayer}.
\bjournal{J. Theoret. Probab.}
\bvolume{15}
\bpages{41--61}.
\bid{doi={10.1023/A:1013885121598}, issn={0894-9840}, mr={1883202}}
\end{barticle}
\bptok{imsref}%
% NOT OUTPUTED:
% issn = 0894-9840
% url = http://dx.doi.org/10.1023/A:1013885121598
% number = 1
% coden = JTPREO
% fjournal = Journal of Theoretical Probability
\endbibitem

%b30 #&#
\bibitem{Zit05}
\begin{barticle}[mr]
\bauthor{\bsnm{{\v{Z}}itkovi{\'c}},~\bfnm{Gordan}\binits{G.}}
(\byear{2005}).
\btitle{Utility maximization with a stochastic clock and an unbounded random endowment}.
\bjournal{Ann. Appl. Probab.}
\bvolume{15}
\bpages{748--777}.
\bid{doi={10.1214/105051604000000738}, issn={1050-5164}, mr={2114989}}
\end{barticle}
\bptok{imsref}%
% NOT OUTPUTED:
% issn = 1050-5164
% url = http://dx.doi.org/10.1214/105051604000000738
% number = 1B
% fjournal = The Annals of Applied Probability
\endbibitem

%b31 #&#
\bibitem{Zit09}
\begin{barticle}[mr]
\bauthor{\bsnm{{\v{Z}}itkovi{\'c}},~\bfnm{Gordan}\binits{G.}}
(\byear{2010}).
\btitle{Convex compactness and its applications}.
\bjournal{Math. Financ. Econ.}
\bvolume{3}
\bpages{1--12}.
\bid{doi={10.1007/s11579-010-0024-z}, issn={1862-9679}, mr={2651515}}
\bptnote{check year}%
\end{barticle}
\bptok{imsref}%
% NOT OUTPUTED:
% issn = 1862-9679
% url = http://dx.doi.org/10.1007/s11579-010-0024-z
% number = 1
% fjournal = Mathematics and Financial Economics
\endbibitem

\end{thebibliography}
\end{document}